\documentclass[prd,twocolumn,superscriptaddress,floatfix,amsmath,amssymb,amsfonts,nofootinbib,longbibliography]{revtex4-2}

\usepackage{tcolorbox}
\usepackage{float} 
\usepackage{scalerel}
\usepackage[normalem]{ulem}
\usepackage[english]{babel}
\usepackage{graphicx}
\usepackage{dcolumn}
\usepackage{bm}
\usepackage{blindtext}
\usepackage{verbatim}
\usepackage{relsize}
\usepackage{mathrsfs}
\usepackage{musicography}
\usepackage{amsmath}
\usepackage{blindtext}
\usepackage{cancel}
\usepackage{physics}
\usepackage{epstopdf}
\usepackage{mathtools}
\usepackage{blindtext}
\usepackage{tensor}
\usepackage{color}
\usepackage[usenames,dvipsnames]{pstricks}
\usepackage{epsfig}
\usepackage{pst-grad} 
\usepackage{pst-plot} 
\usepackage{hyperref}
\usepackage{verbatim}
\usepackage{slashed}
\usepackage{dsfont}
\usepackage[english]{babel}
\usepackage{amsmath,amssymb,amsthm}  
\usepackage{leftindex,enumerate}
\usepackage{enumitem, kantlipsum}  

\renewcommand{\iff}{\Leftrightarrow}

\newcommand{\mf}{\mathsf}

\newcommand{\ii}{\mathrm{i}}

\newcommand{\tc}[1]{\textsc{#1}}

\newcommand{\Ostate}{\zeta}
\newtheorem{thm}{Theorem}

\newtheorem{symm}{Symmetry}
\newtheorem{prop}{Proposition}
\newtheorem{lemma}{Lemma}
\newtheorem{Corollary}{Corollary}

\DeclareMathOperator\erfi{erfi}

\definecolor{goldenrod}{rgb}{0.85, 0.65, 0.13}

\allowdisplaybreaks[1] 

\begin{document}

\title{Interference of communication and field correlations in entanglement harvesting}

\author{Matheus H. Zambianco}
\email{mhzambia@uwaterloo.ca}

\affiliation{Department of Applied Mathematics, University of Waterloo, Waterloo, Ontario, N2L 3G1, Canada}
\affiliation{Institute for Quantum Computing, University of Waterloo, Waterloo, Ontario, N2L 3G1, Canada}

\author{Adam Teixid\'{o}-Bonfill}
\email{adam.teixido-bonfill@uwaterloo.ca}

\affiliation{Department of Applied Mathematics, University of Waterloo, Waterloo, Ontario, N2L 3G1, Canada}
\affiliation{Institute for Quantum Computing, University of Waterloo, Waterloo, Ontario, N2L 3G1, Canada}
\affiliation{Perimeter Institute for Theoretical Physics, Waterloo, Ontario, N2L 2Y5, Canada}

\author{Eduardo Mart\'{i}n-Mart\'{i}nez}
\email{emartinmartinez@uwaterloo.ca}

\affiliation{Department of Applied Mathematics, University of Waterloo, Waterloo, Ontario, N2L 3G1, Canada}
\affiliation{Institute for Quantum Computing, University of Waterloo, Waterloo, Ontario, N2L 3G1, Canada}
\affiliation{Perimeter Institute for Theoretical Physics, Waterloo, Ontario, N2L 2Y5, Canada}

\begin{abstract}
We reveal that the information exchange between particle detectors and their ability to harvest correlations from a quantum field can interfere constructively and destructively. This allows for scenarios where the presence of entanglement in the quantum field is actually detrimental to the process of getting the two detectors entangled. \\
\end{abstract}

\maketitle

\section{Introduction}

There has been a recent push to study entanglement in Quantum Field Theory through particle detector models. Unlike other techniques, particle detector models provide a way to implement localized measurements that make sense in terms of physically implementable settings and are self-regularized and hence free from from spurious divergences. The family of relativistic quantum information protocols known as {\it entanglement harvesting} consists of localized particle detectors that get entangled by extracting pre-existing entanglement from a quantum field. This was first explored by Valentini and Reznik~\cite{Valentini1991, Reznik2003} and further developed later on for different kinds of fields and backgrounds (see, e.g.,~\cite{MenicucciEntanglingPower,HarvestingBHLaura,Pozas-Kerstjens:2015,Pozas2016,boris,freefall,correlation2023Mann, KojiMasahiroPure,Tjoa2020,Cong2019,UDWAds,Henderson2019,Barman2022, HarvestingSuperposed,Liu2021}).


Nonetheless, as first pointed out by \cite{ericksonNew}, not all entanglement acquired by the detectors comes from the correlations in the quantum field. In fact, when the detectors are in causal contact, they can also get entangled via communication. In~\cite{ericksonNew} it was argued that one can split the acquired entanglement acquired by the detectors into two different sources: {\it communication-mediated} entanglement and {\it genuinely harvested} entanglement.

In previous literature it was argued that these two contributions always combine so that there is more entanglement when communication through the field is allowed than in spacelike separation. However, perhaps surprisingly, we show that this is not always the case.

Namely, we find that, in general, the contributions of the communication component and the contribution coming from the harvesting of pre-existing entanglement in the field can actually {\it interfere}. This interference can be constructive, amplifying even more the amount of entanglement acquired by the detectors, but also destructive so that the harvesting contribution and the communication contribution cancel and the detectors fail to get entangled. This phenomenon is not apparent in the setups considered in recent studies in which this splitting is analyzed (such as \cite{ericksonNew,Lensing,lindel2023entanglement}). We show that this phenomenon was not visible because for all the commonly chosen highly symmetric setups (see, e.g., \cite{twist2022, hectorMass, Henderson2019, HarvestingDelocalized, bandlimitedHarv2020, Ng1, SachsMannEdu,entangledDetectors2023,Shahnewaz2023,RalphOlson1,RalphOlson2}) both sources contribute constructively to the entanglement between detectors. In this paper, we provide sufficient conditions, fulfilled by commonly explored scenarios, such that the interference between the two sources is guaranteed to be constructive. 

Then, guided by these conditions we provide examples in Minkowski and De Sitter spacetimes where destructive interference causes the detectors to not get entangled while in causal contact, even though communication alone or genuine harvesting alone would have entangled the detectors. Conversely, we also find scenarios where the constructive interference is stronger than in the setups commonly explored in the literature, further enhancing the entanglement acquired by the detectors.

Our manuscript is organized as follows. In Sec.~\ref{sec:Entanglement Harvesting} we review the protocol of entanglement harvesting, to establish notation and build the necessary tools that will be used later on. In Sec.~\ref{sec:Communication}, we review the splitting of the entanglement acquired by detectors into a communication component and a genuine harvesting component. Sec.~\ref{sec:mathematical_results} presents the main results, namely the conditions that the spacetime, the detectors, and the field Feynmann propagator should meet to guarantee constructive interference. In Sec.~\ref{sec:results}, we show numerical examples in Minkowski and De Sitter spacetimes where these conditions are violated. There, we find that interference can go through the full range from completely destructive to fully constructive. The concluding remarks appear in Sec. \ref{sec:Conclusions}.

\section{Entanglement Harvesting}\label{sec:Entanglement Harvesting}

For the sake of completeness and establishing notation, we review the protocol of entanglement harvesting following, e.g., ~\cite{Pozas-Kerstjens:2015}. Consider two Unurh-DeWitt detectors~\cite{Unruh1976,DeWitt,Unruh-Wald,Jorma,Schlicht}, A and B, following timelike trajectories $\mf z_{\tc{a}}(\tau_{\tc{a}})$ and $\mf z_{\tc{b}}(\tau_{\tc{b}})$ in an $n + 1$ dimensional spacetime $\mathsf{M}$. Here, we assume that $\tau_{\tc{a}}$ and $\tau_{\tc{b}}$ are the respective proper times of the detectors.

The detectors are considered to be two-level quantum systems whose internal Hamiltonians  are given by
\begin{equation}
    \hat{H}_{j} = \Omega_{j} \hat{\sigma}_j^{+} \hat{\sigma}_j^{-},
    \label{H_detectors_AB}
\end{equation}
for $j = \tc{A},\tc{B}$. Here, $\Omega_{j}$ represents the (proper) internal energy gap, while $\hat{\sigma}_j^{+} = |e_j \rangle \langle g_j|$ and $\hat{\sigma}_j^{-}=|g_j \rangle \langle e_j|$ are ladder operators between the respective ground and excited states.

The detectors are coupled to a massless real scalar quantum field $\hat{\phi}(\mf x)$ defined on the spacetime $\mathsf{M}$. The field is assumed to satisfy the Klein-Gordon equation,
\begin{equation}
    (\nabla^{\mu}\nabla_{\mu} - \xi R)\hat{\phi}(\mf x) = 0,
    \label{KG_eq}
\end{equation}
where $\xi$ is a constant, and $R$ is the scalar curvature. Assuming $\{u_{\bf k}(\mf x), u_{\bf k}^{*}(\mf x')\}$ to be a complete, orthonormal set of solutions to Eq.~\eqref{KG_eq}, we can write an explicit expression for the field as
\begin{equation}
    \hat{\phi}(\mf x) = \int{\dd^{n}{\bm k} \  (u_{\bm  k}(\mf x)\hat{a}_{\bm k} + u_{\bm k}(\mf x)^{*}\hat{a}_{\bm k}^{\dagger})}.
    \label{real_scalar_field}
\end{equation}
Here, the ladder operators satisfy the usual bosonic canonical commutation relations, namely
\begin{equation}
    [\hat{a}_{\bm k}, \hat{a}_{\bm p}^{\dagger}] = \delta^{(n)}(\bm k - \bm p).
\end{equation}

We consider the usual covariant prescription for the Unruh-DeWitt detector~\cite{TalesBrunoEdu2020} 
where the interactions happen in a finite region of spacetime according to the following interaction Hamiltonian density\footnote{Notice that this construction assumes that there is a foliation where the spacetime smearing factorizes into a switching and smearing functions $\Lambda_j(\tau_j,\bm x)=\chi_j(\tau_j)F_j(\bm x)$ and that there is a family of frames in which all the spatial points in the support of $F_j(\bm x)$ remain at the same coordinate distance to a ``center of mass'' of the detector $j$. The detector's internal forces keep the detector cohesive. The trajectory of this center of mass $\mathsf{z}_j(\tau_j)$ defines the proper time $\tau_j$ associated with the detector $j$ (i.e., the center of mass is assumed to carry the internal clock of the detector). The details and assumptions that go into this construction can be found in~\cite{TalesBrunoEdu2020,EduTalesBruno2021}, where it is discussed that this kind of construction is reasonable to model systems like atomic probes.} 
\begin{equation}
    \hat{h}_{I}(\mf x) = \lambda \big[\Lambda_{\tc{a}}(\mf x) \hat{\mu}_{\tc{a}}(\tau_\tc{a}) + \Lambda_{\tc{b}}(\mf x) \hat{\mu}_{\tc{b}}(\tau_\tc{b})\big]\hat{\phi}(\mf x).
    \label{H_I_harvesting}
\end{equation}
Here, the detectors' monopole moments are given by
\begin{equation}
    \hat{\mu}_{j}(\tau_j) = e^{\ii \Omega_{j} \tau_{j}} \hat{\sigma}^{+}_{j} + e^{-\ii \Omega_{j} \tau_j} \hat{\sigma}^{-}_{j}.\label{eq:defMonopole}
\end{equation}
and $\Lambda_j({\mf x})$ are the spacetime smearing functions. 

Time evolution in the interaction picture is thus implemented through the operator\footnote{Notice that for spatially smeared detectors there is an ambiguity in the definition of time ordering. However, it was shown in~\cite{EduTalesBruno2021} that if we keep all our predictions to leading order in the perturbative parameter $\lambda$ and choose suitable initial states, then the evolution will not depend on the time ordering chosen, so this ambiguity is irrelevant for this work.}
    \begin{equation}
    \hat{U}_{I} = {\cal T} \exp \left( -\ii\int \dd V  \hat{h}_{\tc{I}}(\mf x)\right),
    \label{U_I}
\end{equation}
where $\dd V$ is the invariant spacetime volume element, that for a given choice of coordinates $(t,\bm x)$ takes the form $\dd V= \dd t \dd^{n}\bm x\sqrt{-g}$, where $g$ is the determinant of the metric tensor. 

As it is common in entanglement harvesting protocols, we can assume that the initial state of the field is a zero-mean Gaussian state $\hat{\rho}_{\phi, 0}$ (like for example the Minkowski vacuum in flat spacetime and the conformal vacuum in DeSitter), whereas the detectors can be set in arbitrary pure states $\ket{\psi_{\tc{a}}}$ and $\ket{\psi_{\tc{b}}}
$. The initial state of the full system is then written as 
\begin{equation}
    \hat{\rho}_{0} = |\psi_{\tc{a}} \rangle \langle \psi_\tc{a}| \otimes |\psi_\tc{b} \rangle \langle \psi_\tc{b}| \otimes \hat{\rho}_{\phi, 0}.
    \label{rho_0_harvesting}
\end{equation}

With this choice, notice that the detectors' initial state has no correlations whatsoever. Nonetheless, after time evolving the state $\hat{\rho}_{0}$ with Eq.~\eqref{U_I}, the detectors will generically evolve to a mutually entangled state. This is true even if the detectors are spacelike separated. This is so because the detectors have harvested pre-existing entanglement in the field's initial state. This fact gives the name of {\it entanglement harvesting} to this phenomenon.

The state of the detectors after they couple to the field is given by $\hat{\rho}_{\tc{ab}} =  \Tr_{\phi}[\hat{U}_{I} \hat{\rho}_{0} \hat{U}_{I}^{\dagger}]$. In order to express this state in matrix form, we consider states $\ket{\Ostate_\tc{a}}$ and $\ket{\Ostate_\tc{b}}$ so that ${\cal B} = \{\ket{\psi_\textsc{a} \psi_{\textsc{b}}}, \ket{\psi_\textsc{a} \Ostate_\tc{b}}, \ket{\Ostate_\tc{a} \psi_{\textsc{b}}}, \ket{\Ostate_\tc{a} \Ostate_\tc{b}}\}$ is an orthonormal basis. In terms of the previously defined ground and excited states of the detectors, we can write
\begin{align}
    \nonumber\ket{g_{{j}}} &= \cos{\alpha_{{j}}} \ket{\psi_{{j}}} - e^{\ii \beta_{{j}}} \sin{\alpha_{{j}}} \ket{\Ostate_j}, \\
    \ket{e_{{j}}} &= e^{-\ii \beta_{{j}}}\sin{\alpha_{{j}}} \ket{\psi_{{j}}} +  \cos{\alpha_{j}} \ket{\Ostate_j},\label{eq:states}
\end{align}
for some $\alpha_{j}, \beta_{j} \in [0, 2\pi]$. Thus, in the basis ${\cal B}$ we have
\begin{align}
    \hat{\rho}_{\textsc{ab}} &= \left[
\begin{array}{cccc}
1-\mathcal{L}_\textsc{aa}-\mathcal{L}_\textsc{bb} & \mathcal{X}^* & \mathcal{Y}^* & \mathcal{M}^\ast \\
\mathcal{X} & \mathcal{L}_\textsc{bb} & \mathcal{L}_\textsc{ab}^*  & 0 \\
\mathcal{Y} & \mathcal{L}_\textsc{ab} & \mathcal{L}_\textsc{aa} &0\\
\mathcal{M} & 0 & 0 & 0 \\
\end{array}\right] +O(\lambda^4) \;,
\label{rho_AB_harvesting}
\end{align}
with 
\begin{align} 
    \mathcal{L}_{{ij}}= \big[ &\cos^2{\alpha_{{i}}}\cos^2{\alpha_{{j}}} \,L_{{ij}}(\Omega_{i}, \Omega_{j})\nonumber\\
    &- \cos^2{\alpha_{{i}}}\sin^2{\alpha_{{j}}} \,e^{-2\ii \beta_{{j}}} \nonumber \,L_{{ij}}(\Omega_{i}, -\Omega_{j})\\ & - \sin^2{\alpha_{{i}}}\cos^2{\alpha_{{j}}}\,e^{2\ii \beta_{{i}}} \,L_{{ij}}(-\Omega_{i}, \Omega_{j}) \nonumber\\
    & + \sin^2{\alpha_{{i}}}\sin^2{\alpha_{{j}}}\,e^{2\ii (\beta_{{i}} - \beta_{{j}} )} \,L_{{ij}}(-\Omega_{i}, -\Omega_{j}) \label{eq:LijGen} \big],\\
    \mathcal{M}= \big[&\cos^2{\alpha_{\tc{a}}}\cos^2{\alpha_{\tc{b}}} \,M(\Omega_{\tc{a}}, \Omega_{\tc{b}})\nonumber\\
    &- \cos^2{\alpha_{\tc{a}}}\sin^2{\alpha_{\tc{b}}}\, e^{2\ii \beta_{\tc{b}}} \, M(\Omega_{\tc{a}}, -\Omega_{\tc{b}})\ \nonumber \\ 
    &- \sin^2{\alpha_{\tc{a}}}\cos^2{\alpha_{\tc{b}}}\,e^{2\ii \beta_{\tc{a}}} \,M(-\Omega_{\tc{a}}, \Omega_{\tc{b}})\ \nonumber\\
    &+ \sin^2{\alpha_{\tc{a}}}\sin^2{\alpha_{\tc{b}}}\,e^{2\ii (\beta_{\tc{a}} + \beta_{\tc{b}} )} \,M(-\Omega_{\tc{a}}, -\Omega_{\tc{b}})\big] \label{eq:MGen},
\end{align}
where 
\begin{align}
    L_{ij}(\Omega_{i}, \Omega_{j})  =  \lambda^2 \int\dd V \dd V'& e^{\ii (\Omega_{i} \tau_i-\Omega_{j}\tau_j')} \nonumber\\ \times 
    &\Lambda_{i}(\mf x) \Lambda_{j}(\mf x')W(\mf x', \mf x),\label{Lij_harvesting} \\
    M(\Omega_{\tc{a}}, \Omega_{\tc{b}}) = -\lambda^2 \int \dd V \dd V' & e^ {\ii( \Omega_{\tc{a}}\tau_\tc{a}+\Omega_{\tc{b}} \tau_\tc{b}')} \nonumber \\ \times & \Lambda_{\tc{a}}(\mf x)\Lambda_{\tc{b}}(\mf x')G_{\tc{f}}(\mf x, \mf x').
    \label{M_harvesting}
\end{align}
The terms ${\cal L}_{ij}$ and ${\cal M}$ are usually called {\it local noise} and {\it non-local correlations}, respectively. The expressions for ${\cal X}$ and ${\cal Y}$ are cumbersome and, since we will not use them here, we refer the interested reader to reference \cite{statedependenceHector}. In equations~\eqref{Lij_harvesting} and~\eqref{M_harvesting}, $W(\mf x, \mf x')$ is the {\it Wightman function},
\begin{equation}
    W(\mf x, \mf x') = \Tr[\hat{\phi}(\mf x)\hat{\phi}(\mf x') \hat{\rho}_{\phi, 0}],
\end{equation}
whereas $G_{\tc{f}}(\mf x, \mf x')$ represents the {\it Feynman propagator}, which can be written as
\begin{equation}
\begin{aligned}
    G_{\tc{f}}(\mf x, \mf x') &= \Tr[{\cal T}\hat{\phi}(\mf x)  \hat{\phi}(\mf x') \hat{\rho}_{\phi,0}]  \\&= \Theta(t - t') W(\mf x, \mf x') + 
    \Theta(t' - t) W(\mf x', \mf x). \label{eq:defGF}
\end{aligned}
\end{equation}
Here, $\Theta(t)$ is the Heaviside step function, and $t$ describes any time coordinate.

The amount of entanglement between the detectors can be quantified by the {\it negativity}~\cite{VidalNegativity}. For a system of two qubits, Negativity is a faithful entanglement measure, and it is defined as the absolute sum of the {\it negative eigenvalues} of the partially transposed density matrix of the two qubits $\hat{\rho}^{t_{\tc{a}}}$. In particular, for the state $\hat{\rho}_{\tc{ab}}$ described by Eq.~\eqref{rho_AB_harvesting} we have~\cite{Pozas-Kerstjens:2015}
\begin{equation}
   {\cal N}=  \max\{0, {\cal V}\}+ {\cal O}(\lambda ^{4}),
\end{equation}
with
\begin{equation}
    {\cal V} = \sqrt{|{\cal M}|^{2} +\left(\frac{{\cal L}_{\tc{aa}}
    - {\cal L}_{\tc{bb}}}{2}\right)^2 } - \frac{{\cal L}_{\tc{aa}} + {\cal L}_{\tc{bb}}}{2}.
    \label{V_negativity_general}
\end{equation}
According to expression~\eqref{V_negativity_general}, we can see that, in general, entanglement emerges as a competition between the non-local term ${\cal M}$ and the local noise terms (${\cal L}_{\tc{aa}}$ and ${\cal L}_{\tc{bb}}$). For the particular case where \mbox{$\mathcal{L}_{\tc{aa}} = \mathcal{L}_{\tc{bb}} = \mathcal{L}$} (such as identical inertial detectors in Minkowski spacetime), the negativity reads
\begin{equation}
    {\cal N} = \max\{0, |{\cal M}| - {\cal L}\} + {\cal O}(\lambda ^{4}).
\end{equation}
When the detectors are in causal contact, it is possible to trace back the entanglement created between the detectors to two different sources~\cite{ericksonNew}. On the one hand, the detectors can get entangled via communication through the field ({\it signaling}) and, on the other hand, the detectors can get entangled by extracting the entanglement already present~\cite{vacuumEntanglement,vacuumBell} in the field state ({\it harvesting}). In the next section, we are going to review how one can distinguish between these two different contributions.

\section{Communication vs genuine harvesting}\label{sec:Communication}

When two detectors in spacelike separation get entangled through their interaction with the field, it is clear that this entanglement has to come from the pre-existing correlations in the field itself. However, when they are in causal contact, the detectors could very well be gaining entanglement both through genuine harvesting from the field state and also through the direct exchange of information through the field. The techniques to separate these two contributions were first laid out in \cite{ericksonNew}, where the authors use the state dependence of the symmetric (real) and anti-symmetric (imaginary) parts of the Wightman function to split the negativity into signaling and genuine harvesting contributions. Following the same notation as in~\cite{ericksonNew}, we are going to denote by ${\cal N}^{-}$ the negativity computed if only the signaling contribution is kept in the $\mathcal{M}$ term, and ${\cal N}^{+}$ the negativity if we only keep the symmetric contribution of the Wightman function for the $\mathcal{M}$ term. Then, determining if entanglement is harvested or if it is acquired through communication is done by comparing the total amount of entanglement $\mathcal{N}$ and the signaling contribution $\mathcal{N}^-$.

Concretely, to define ${\cal N}^{\pm}$ we first look at the state dependence of the symmetric and anti-symmetric parts of the Wightman function. We separate the Wightman function as follows:
\begin{equation}
    W(\mf x, \mf x') = W^{+}(\mf x, \mf x) + W^{-}(\mf x , \mf x'),
    \label{Wightman_real_imag}
\end{equation}
where the symmetric and the anti-symmetric parts are respectively given by
\begin{equation}
    W^{+}(\mf x , \mf x') = \frac{1}{2}\Tr(\{\hat{\phi}(\mf x), \hat{\phi}(\mf x')\}\hat{\rho}_{\phi})
\end{equation}
and
\begin{equation}
    W^{-}(\mf x , \mf x') = \frac{1}{2}\Tr([\hat{\phi}(\mf x), \hat{\phi}(\mf x')] \hat{\rho}_{\phi}).
\end{equation}
The commutator $[\hat{\phi}(\mf{x}), \hat{\phi}(\mf{x}')]$ is a state-independent multiple of the identity. The state independence of $W^{-}(\mf{x}, \mf{x}')$ is key to separating the two contributions: any entanglement whose origin can be traced back to this component of the Wightman function cannot be associated to pre-existing entanglement in the field, since one could consider a state of the field with no entanglement and yet this contribution would be unchanged. Consequently, it must be concluded that the anti-symmetric part of the Wightman function is not associated with true entanglement harvesting but rather contributes only to signaling.

Using the decomposition of the Wightman in Eq.~\eqref{Wightman_real_imag}, we can formally define the symmetric and anti-symmetric parts of the Feynman propagator as
\begin{equation}
    G_{\tc{f}}^{\pm}(\mf x, \mf x') = \Theta(t - t') W^{\pm}(\mf x, \mf x') + 
    \Theta(t' - t) W^{\pm}(\mf x', \mf x).
\end{equation}
Then, we can split the non-local term ${\cal M}$ into two contributions, namely
\begin{equation}
    {\cal M} = {\cal M}^{+} + {\cal M}^{-}\label{eq:sumMs},
\end{equation}
the components ${\cal M}^{\pm}$ can be defined by simply changing the Feynman propagator in Eq.~\eqref{M_harvesting} by the corresponding terms $G_{\tc{f}}^{\pm}(\mf x, \mf x')$. Finally, we write
\begin{equation}
    {\cal V}^{\pm} =  \sqrt{|{\cal M}^{\pm}|^{2} +\left(\frac{{\cal L}_{\tc{aa}}
    - {\cal L}_{\tc{bb}}}{2}\right)^2 } - \frac{{\cal L}_{\tc{aa}} + {\cal L}_{\tc{bb}}}{2}.
    \label{eq:V_pm}
\end{equation}
so that the {\it harvested negativity} ${\cal N}^{+}$ and the {\it communication-assisted negativity} ${\cal N}^{-}$ can be expressed as\footnote{A small remark to avoid confusion due to the names given to the different contributions to negativity: notice that $\mathcal{M}^-$ can still be non-zero and $\mathcal{N}^-=0$ if the local noise terms conspired appropriately (or it could even happen that $\mathcal{N}\neq 0$ while \mbox{$\mathcal{N}^\pm=0$}, a case that will appear in the plots of the examples of the section \ref{sec:results}
). 
This means that there can be some communication-assisted harvesting even if $\mathcal{N}^-=0$. However, what is true is that even in such a case the pre-existing field correlations are responsible for the entanglement between the detectors. If field correlations were not there, there would be no entanglement in the detectors at all. }
\begin{equation}
    {\cal N}^{\pm} = \max\{{\cal V}^{\pm}, 0\} + {\cal O}(\lambda ^{4}).
\end{equation}

\section{General results for guaranteeing constructive interference}
\label{sec:mathematical_results}

It is commonly observed in previous literature that the acquired entanglement ${\cal N}$ is larger than both ${\cal N}^{\pm}$, which indicates that communication and pre-existing field correlations work together to entangle the detectors. However, this is not the case in general. There are situations where communication and field correlations can ``destructively interfere'', resulting in small $\mathcal{N}$ even if both ${\cal N}^{\pm}$ are relatively large. In this section, we will explore this phenomenon by describing all the ways in which the contributions from communication and pre-existing field correlations can add up to build entanglement. Then, we explore a set of transformations of the harvesting setup that turn scenarios where communication and genuine harvesting contributions to $|\mathcal{M}|$ cancel each other into scenarios where they collaborate. After this, we provide a set of (easy to check) sufficient conditions under which communication and harvesting are always collaborating constructively.

For this exploration, we will focus on the study of the correlation term $|\mathcal M|$ which always contributes positively to $\mathcal N$. The current section specifically presents how communication (through $\mathcal M^-$) and field correlations (through $\mathcal M^+$) determine $|\mathcal M|$ and therefore $\mathcal N$. The local noise terms, $\mathcal L_\tc{aa}$ and $\mathcal L_\tc{bb}$, of course, also affect entanglement, however, they are local to each detector and hence independent of communication and correlations acquired of the detectors, and thus are ignored in the present section. 

Starting from the most general scenario, $|\mathcal M|$ can be small even if $|\mathcal{M}^+|$ and $|\mathcal{M}^-|$ are relatively large, because
\begin{align}
    |\mathcal M|^2 &= |\mathcal M^+ + \mathcal M^-|^2 \nonumber\\
    &= |\mathcal M^+|^2 + |\mathcal M^-|^2 + 2|\mathcal M^+||\mathcal M^-|\cos \Delta\gamma.\label{eq:complexsum}
\end{align}
Here, $\Delta\gamma$ is the relative phase between the complex numbers $\mathcal M^+$ and $\mathcal M^-$. This phase $\Delta\gamma$ controls how communication ($\mathcal M^-$) interferes with field correlations ($\mathcal M^+$). The interference can be \textit{fully destructive}, making $|\mathcal M|$ the smallest,
\begin{equation}
    |\Delta\gamma| = \pi \iff |\mathcal M| = ||\mathcal M^+|-|\mathcal M^-||,
\end{equation}
\textit{fully constructive}, making $|\mathcal M|$ the largest,
\begin{equation}
    |\Delta\gamma| = 0 \iff |\mathcal M| = |\mathcal M^+|+ |\mathcal M^-|,
\end{equation}
or exactly \textit{in between}, with ${\cal M}^+$ and ${\cal M}^-$ so that  
\begin{equation}
    |\Delta\gamma| = \pi/2 \iff |{\cal M}| = \sqrt{|{\cal M}^+|^2+|{\cal M}^-|^2}.
    \label{eq:abs_M_sqrt}
\end{equation}
In the two latter cases, communication and field correlations collaborate to increase $|{\cal M}|$ and thus $\mathcal N$. As we will see, the last relation, Eq.~\eqref{eq:abs_M_sqrt}, turns out to be commonly fulfilled for the simplest choices of UDW detector model, as will be explored in Proposition \ref{prop:switchings}

Next, we explore how to transform a setup from a destructive interference to a constructive interference scenario. Concretely, we will show how to transform the parameters of the harvesting setup to change the sign $\cos \Delta\gamma$ in Eq~\eqref{eq:complexsum}. In order to find such transformation, it is convenient to denote as $\mathcal{M}_{G_\tc{f}^*}$ the term $\mathcal{M}$ after complex conjugating the propagator $G_\tc{f}$. Then,
\begin{align}
    \mathcal{M} = \mathcal M^+ +\mathcal M^-,\ \mathcal{M}_{G_\tc{f}^*} = \mathcal M^+ -\mathcal M^-. 
    \label{eq:MGF_star}
\end{align}
Consequently,
\begin{equation}
    \mathcal{M}^\pm_{G_\tc{f}^*}=\pm\mathcal{M}^{\pm}, \label{eq:Mpm_relation}
\end{equation}
which implies
\begin{equation}
    |\mathcal M_{G_\tc{f}^*}|^2 = |\mathcal M^+|^2 + |\mathcal M^-|^2 - 2|\mathcal M^+||\mathcal M^-|\cos \Delta\gamma\label{eq:interferenceMGF},
\end{equation}
where $\Delta\gamma$ is the relative phase between the original terms $\mathcal M^{+}$ and $\mathcal M^-$. Observe that, as we were searching for, the term with $\cos \Delta\gamma$, which controls the interference, has the opposite sign in $|\mathcal{M}_{G_\tc{f}^*}|$ compared to in $|\mathcal{M}|$. Meanwhile, the communication and field correlations contributions retain the same size, $|\mathcal{M}^\pm_{G_\tc{f}^*}|=|\mathcal{M}^{\pm}|$.

Therefore, if we find an operation that conjugates $G_\tc{f}$, we could use it to turn destructive interference into constructive, as previously described. The next proposition shows how one can achieve that by changing the initial state of the particle detectors or, equivalently, transforming the parameters of the particle detector model.
\begin{prop} \label{prop:transformation}
Consider an operation that transforms arbitrary states of the detectors $j=\text{A, B}$ as follows
\begin{align}
    & \cos \alpha_{j} \ket{g_j} + e^{\ii \beta_{j}} \sin \alpha_{j} \ket{e_j}  \nonumber \\
    &  \to \sin \alpha_{j} \ket{g_j} + e^{\ii \beta_{j}} \cos \alpha_{j} \ket{e_j}.\label{eq:transfrom_states}
\end{align}
Let $\widetilde{\mathcal M}$ be $\mathcal M$ after performing this transformation on the initial state of the detectors. Then,
\begin{align}
    &|\widetilde{\mathcal M}|^2=|\mathcal M^+|^2 + |\mathcal M^-|^2 - 2|\mathcal M^+||\mathcal M^-|\cos \Delta\gamma,\nonumber\\
    &|\widetilde{\mathcal M}^\pm|=|\mathcal{M}^{\pm}|.
\end{align} 
Therefore, swapping the populations in the ground and excited state (while keeping the relative phase between them) changes the sign of $\cos \Delta\gamma $. Thus, this swap turns any destructive interference into constructive, while keeping the communication and genuine harvesting contributions unchanged.

Equivalently, the transformation in Eq.~\eqref{eq:transfrom_states} can be implemented by reversing the energy gaps \mbox{$\Omega_j\to-\Omega_j$} and the relative phases $\beta_j\to-\beta_j$.


\begin{proof}

To prove the proposition, use Eq.~\eqref{eq:MGen} to see that
\begin{align}
    (\mathcal{M}_{G_\tc{f}^*})^*=&\big[\cos^2{\alpha_{\tc{a}}}\cos^2{\alpha_{\tc{b}}} \,M(-\Omega_{\tc{a}}, -\Omega_{\tc{b}})\nonumber\\
    &\phantom{\big[}- \cos^2{\alpha_{\tc{a}}}\sin^2{\alpha_{\tc{b}}}\, e^{-2\ii \beta_{\tc{b}}} \, M(-\Omega_{\tc{a}}, \Omega_{\tc{b}})\ \nonumber \\ 
    &\phantom{\big[}- \sin^2{\alpha_{\tc{a}}}\cos^2{\alpha_{\tc{b}}}\,e^{-2\ii \beta_{\tc{a}}} \,M(\Omega_{\tc{a}},- \Omega_{\tc{b}})\ \nonumber\\
    &\phantom{\big[}+ \sin^2{\alpha_{\tc{a}}}\sin^2{\alpha_{\tc{b}}}\,e^{-2\ii (\beta_{\tc{a}} + \beta_{\tc{b}} )} \,M(\Omega_{\tc{a}}, \Omega_{\tc{b}})\big]\nonumber\\
    =&\widetilde{\mathcal M}.\label{eq:Mtilde_iskinda_MGF}
\end{align}
Here we used that
\begin{equation}
    \big(M_{G_\tc{f}^*}(\Omega_{\tc{a}}, \Omega_{\tc{b}})\big)^*=M(-\Omega_{\tc{a}}, -\Omega_{\tc{b}}),
\end{equation}
which in turn follows from Eq.~\eqref{M_harvesting} and $\Lambda_j=\Lambda_j^*$. 

Finally, substituting
$\widetilde{\mathcal M}=\mathcal{M}_{G_\tc{f}^*}^*$ into Eqs.~\eqref{eq:Mpm_relation} and~\eqref{eq:interferenceMGF} provides the identities stated in the proposition. 
\end{proof}
\end{prop}
Notice that according to this result, the stronger the destructive interference is in $\mathcal M$, the stronger the corresponding constructive interference is in $\widetilde{\mathcal M}$. 

Now, we are ready to explore the conditions that enforce \mbox{$\cos \Delta\gamma = 0$}. 

\begin{prop} \label{prop:equalMMGF} Let $\mathcal{M}_{G_\tc{f}^*}$ be $\mathcal{M}$ after complex conjugating $G_\tc{f}$, as in Eq.~\eqref{eq:MGF_star}. Then,
\begin{align}
        |\mathcal M| = |\mathcal{M}_{G_\tc{f}^*}| &\iff |{\cal M}| = \sqrt{|{\cal M}^+|^2+|{\cal M}^-|^2}\label{eq:orthogonalNew}\\
        &\iff \cos\Delta\gamma=0.\nonumber
\end{align}
\begin{proof}
    The proof follows directly from comparing equations~\eqref{eq:complexsum} and~\eqref{eq:interferenceMGF}.
\end{proof}
\end{prop}
Therefore, $|\mathcal M| = |\mathcal{M}_{G_\tc{f}^*}|$ is a \textbf{sufficient condition} for both the communication ${\cal M}^{-}$ and field correlations ${\cal M}^{+}$ to positively contribute to negativity. Notice that, because of Eq.~\eqref{eq:Mtilde_iskinda_MGF}, the result of Proposition~\ref{prop:equalMMGF} also holds if we use $\widetilde{\mathcal M}$ instead of $\mathcal {M}_{G_\tc{f}^*}$. Therefore, if $|\mathcal{M}|$ does not change under transforming the initial states as in Eq.~\eqref{eq:transfrom_states}, then there cannot be destructive interference. Furthermore, whenever $|\mathcal M| = |\mathcal{M}_{G_\tc{f}^*}|$, interference is locked in the middle ground with $\cos\Delta\gamma=0$, and the contributions ${\cal N}^{-}$ and ${\cal N}^{+}$ to the entanglement between the detectors satisfy
\begin{align}
    \mathcal N \geq \sqrt{({\cal N}^+)^2+{(\cal N}^-)^2}\;\Rightarrow{\cal N} \ge {\cal N}^{-},\;
      {\cal N} \ge {\cal N}^{+},\label{conditions_negativity_beats_all}
\end{align}
as proven in Appendix~\ref{apx:inequality_negativity}.

Next, we will provide a framework to find symmetries of the harvesting setup that enforce the condition $|\mathcal M| = |\mathcal{M}_{G_\tc{f}^*}|$. This framework is based on performing coordinate transformations in the integral for $\mathcal M$. To keep this method as general as possible and simultaneously simplify the derivation, we use the order $\lambda^2$ contribution to $\hat\rho_\tc{ab}$, denoted by $\hat\rho^{(2)}_\tc{ab}$ (the expression is provided in reference \cite{statedependenceHector}). Working with the basis generated by the tensor product of the general states $\ket{\psi_j}$ and $\ket{\Ostate_j}$, as defined before in Eq.~\eqref{eq:states}, we can write
\begin{align}
    \mathcal M &= \mel{\Ostate_\tc{a}\chi_\tc{b}}{\hat\rho^{(2)}_\tc{ab}}{\psi_\tc{a}\psi_\tc{b}},\nonumber\\
    &= -\lambda^2 \int{\dd V \dd V' P_\tc{a}(\mf x)P_\tc{b}(\mf x')G_{\tc{f}}(\mf x, \mf x')}. 
\end{align}
where we have defined\footnote{Here we write $\tau_j(\mf{x})$ to represent the value of $\tau_j$ when the center of mass of detector $j$ is in the same leaf (of the Fermi-Walker foliation where the spacetime smearing factorizes as mentioned in footnote 1) as the spacetime point $\mathsf{x}$.}
\begin{equation}
     P_j(\mf x) =\Lambda_j(\mf x)\mel{\Ostate_j}{\hat\mu_j(\tau_j(\mf x))}{\psi_j}.\label{eq:defPs}
\end{equation}
Notice that the terms $P_j(\mf x)$ conveniently group the contributions to $\mathcal M$ due to each detector. The condition $|\mathcal M| = |\mathcal{M}_{G_\tc{f}^*}|$ thus becomes equivalent to
\begin{align}
    &\bigg|\int{\dd V \dd V' P_\tc{a}(\mf x)P_\tc{b}(\mf x')G_{\tc{f}}(\mf x, \mf x')}\bigg| \nonumber\\
    &= \bigg|\int{\dd V \dd V' P^*_\tc{a}(\mf x)P^*_\tc{b}(\mf x')G_{\tc{f}}(\mf x, \mf x')}\bigg|.\label{eq:halfway_condition}
\end{align}
Now, consider a map $\bm T$ taking a pair of points $\mf x, \mf x' \in \mathsf M$ to another pair as follows,
\begin{equation}
    \bm T (\mf x, \mf x') = (T_1(\mf x, \mf x'),T_2(\mf x, \mf x')). \label{eq:Tmap_definition}
\end{equation}

To see the effect of this map on pairs of coordinates, consider a pair of charts $(\varphi,\varphi')$ mapping pairs of spacetime points $(\mf x,\mf x')$ to the coordinates $(x^\mu,x'^\mu)\in\mathbb{R}^{n + 1}$. We can see that $\bm T$ induces a map between the coordinates of pairs of events  $(y^{\mu}, {y'}^{\mu}) \to (x^{\mu}, {x'}^{\mu})$ through
\begin{equation}
    (\varphi,\varphi')\circ \bm T^{-1}\circ(\varphi^{-1},\varphi'^{-1})\label{AdamTmap}, 
\end{equation} 
i.e., through the composition of $\bm T^{-1}$ with the chart pair $(\varphi,\varphi')$. Specifically, the coordinates $(y^{\mu}, {y'}^{\mu})$ of the pair of points $(\mf y, \mf y')$ are sent to the the coordinates $(x^{\mu}, {x'}^{\mu})$ of the pair of points $(T_1^{-1}(\mf y, \mf y'),T_2^{-1}(\mf y, \mf y'))$. Then, for integrals like the ones in Eq.~\eqref{eq:halfway_condition}, we can use the map~\eqref{AdamTmap} as a ``change of variables'', as follows\footnote{Notice that when we replace a function of points in spacetime $A(\mf x,\mf x')$ by $A(x^\mu,{x'}^\mu)$ we are abusing notation and what we actually mean is $A(\mf x,\mf x')=A( \varphi^{-1}(x^\mu),\varphi'^{-1}({x'}^\mu))\equiv A( x^\mu,{x'}^\mu)$.}
\begin{align}
    \mathcal A &= \int \dd V \dd V' A(\mf y, \mf y')\nonumber\\*
    &=\int \dd^{n+1}y^{\mu}\, \dd^{n+1} {y'}^{\mu} \sqrt{g(y^{\mu}) g'({y'}^{\mu})}A(y^{\mu},  {y'}^{\mu})\nonumber\\*
    &=\int \dd^{n+1}x^{\mu}\, \dd^{n+1}{x'}^{\mu} \sqrt{g\big( T_1(x^{\mu}, {x'}^{\mu})\big)\, g' \big(T_2(x^{\mu}, {x'}^{\mu})\big)} \nonumber\\*
    &\qquad \times |J(x^{\mu}, {x'}^{\mu})| A \big(\bm T (x^{\mu}, {x'}^{\mu})\big),\label{eq:complicated}
\end{align}
where $g$ and $g'$ are the determinant of the metric in the coordinates defined by $\varphi$ and $\varphi'$ respectively. Moreover,
\begin{equation}
    J(x^{\mu}, {x'}^{\mu}) = \det[\frac{\partial(y^{\mu}, {y'}^{\mu})}{\partial(x^{\nu}, {x'}^{\nu})}]\label{eq:jacobian_definition}
\end{equation}
is the determinant of the Jacobian matrix of the change of variables. For convenience, let us define 
\begin{equation}
    \mathfrak{j}(x^{\mu}, {x'}^{\mu})=\frac{\sqrt{g\big(T_1(x^{\mu}, {x'}^{\mu})\big)\, g' \big(T_2(x^{\mu}, {x'}^{\mu})\big)} }{\sqrt{g(x^{\mu}) g'({x'}^{\mu})}}|J(x^{\mu}, {x'}^{\mu})|,
\end{equation}
Notice that $\mathfrak{j}(x^{\mu}, {x'}^{\mu})$ is invariant under changes of coordinates. Therefore, we can write Eq.~\eqref{eq:complicated} in an explicitly coordinate independent way,
\begin{equation}
    \mathcal A = \int \dd V \dd V'\, \mathfrak{j}(\mf x, \mf x') A \big(\bm T (\mf x, \mf x')\big).
\end{equation}
Now, consider a different integral,
\begin{align}
    \mathcal B &= \int \dd V \dd V' B(\mf x, \mf x').
\end{align}
Then, observe that in order to impose the condition \mbox{$|\mathcal A| = |\mathcal B|$} it is sufficient to find a constant angle $\nu$ such that 
\begin{align}
    \mathfrak{j}(\mf x, \mf x') A \big(\bm T (\mf x, \mf x')\big)= e^{\ii\nu} B(\mf x, \mf x').
\end{align}
Applying this result to our condition in Eq.~\eqref{eq:halfway_condition}, we obtain the following theorem.

\begin{thm}\label{thm:changeOfVariables}
    Consider any transformation $\bm T=(T_1,T_2)$ that acts on two copies of the spacetime $\mathsf{M}$, as in equation~\eqref{eq:Tmap_definition}. Let $P_{\tc{a}}$ and $P_{\tc{b}}$ be given by equation~\eqref{eq:defPs}.
    
    Assume there is a constant angle $\nu$ such that
    \begin{align}
        &\frac{P_\tc{a}^*\big(T_1(\mf x, \mf x')\big)P_\tc{b}^*\big(T_2(\mf x, \mf x')\big)G_{\tc{f}}\big(\bm T (\mf x, \mf x')\big)\mathfrak{j}(\mf x, \mf x')}{P_\tc{a}(\mf x)P_\tc{b}(\mf x')G_{\tc{f}}(\mf x, \mf x')}= e^{\ii\nu}.
        \label{eq_theorem1}
    \end{align}
    Then, $|\mathcal{M}| = |\mathcal M_{G_\tc{f}^*}|$, and therefore by Proposition~\ref{prop:equalMMGF},
    \begin{equation}
        |{\cal M}| = \sqrt{|{\cal M}^+|^2+|{\cal M}^-|^2}.
    \end{equation}
    Which means that \mbox{$\cos \Delta \gamma =0$} in Eq.~\eqref{eq:complexsum} and that, in particular, the contributions of communication and pre-existing field correlations cannot interfere destructively.    
\end{thm}

It turns out that we can simplify the condition expressed by Eq.~\eqref{eq_theorem1} if, for a given spacetime and field configuration, $\mathfrak{j}(\mf x, \mf x')=1$ and the map $\bm T$ preserves the Feynman propagator. This simplification is detailed in the following Corollary.

\begin{Corollary}\label{cor:symmetriesT}
    Assume that the map $\bm T =(T_1,T_2)$ satisfies 
    \begin{align}
        &g\big(T_1(x^{\mu}, {x'}^{\mu})\big)\, g' \big(T_2(x^{\mu}, {x'}^{\mu})\big) |J(x^{\mu}, {x'}^{\mu})|^2=g(x^{\mu}) g'({x'}^{\mu}), \nonumber\\
        &G_{\tc{f}}\big(\bm T (\mf x, \mf x')\big)=G_{\tc{f}}(\mf x, \mf x'),\label{eq:symmgGF}
    \end{align}
    Then, the condition
    \begin{align}
        P_\tc{a}^*\big(T_1(\mf x, \mf x')\big)P_\tc{b}^*\big(T_2(\mf x, \mf x')\big) = e^{\ii\nu}P_\tc{a}(\mf x)P_\tc{b}(\mf x'),\label{eq:conjSymmP}
    \end{align}
    with $\nu$ some constant, becomes sufficient to have
    \begin{equation}
        |{\cal M}| = \sqrt{|{\cal M}^+|^2+|{\cal M}^-|^2}.
        \label{eq:no_interference}
    \end{equation}
\end{Corollary}

In order to apply this result to physical situations, it will be helpful to further simplify this corollary when certain symmetries are present in the detector dynamics. Several of these possible symmetries are explored in Appendix~\ref{apx:ExaSym}. From there, we derive that in many cases, which happen to frequently appear in the literature, communication and field correlations must ``collaborate'' to generate entanglement between detectors. Moreover, this collaboration happens in the very specific way described by Eq.~\eqref{eq:no_interference}. Concretely, in Appendix~\ref{apx:ExaSym} we provide conditions on the particle detectors, field, and spacetime of entanglement harvesting setups so that Eq.~\eqref{eq:no_interference} holds. It is important to have in mind that these conditions are only {\it sufficient} to force $\cos \Delta\gamma=0$ and, in turn, to rule out destructive interference. Nonetheless, as we shall see in the next section, it is easy to find setups even with fully constructive interference ($\Delta \gamma=0$) or fully destructive interference ($\Delta \gamma=\pi$) just by violating these conditions. Therefore, even if these conditions are not necessary, they seem to be present in all cases that we know of where Eq.~\eqref{eq:no_interference} holds. For conciseness, the following proposition presents a simplified version of these conditions, which will guide the exploration of harvesting setups in the next section.

\begin{tcolorbox}[colframe=blue!20,
colback=blue!5]
\begin{prop}\label{prop:switchings}
Consider entanglement harvesting setups where
\begin{enumerate}
    \item Spacetime is flat.
    \item The field is initially prepared in the Minkowski vacuum, and the detectors start both in their ground state or both in their excited state. 
    \item The detectors are inertial and comoving.
    \item Spacetime smearings are separable in the detectors' comoving frame, i.e., \mbox{$\Lambda_j(t,\bm x)= \chi_j(t)F_j(\bm x)$}.
    \item The proper energy gaps are equal, $\Omega_\tc{a}= \Omega_\tc{b}$.
    \item In the comoving frame, for some time $t_{\tc{r}}$, it is satisfied that either
    \begin{align}
        &\chi_j(t+t_{\tc{r}})=\chi_j(t_{\tc{r}}-t),\ j=A,B,\label{eq:reverse}\\
        &\text{or }\;\chi_\tc{a}(t+t_{\tc{r}})=\chi_\tc{b}(t-t_{\tc{r}}).\label{eq:swapandreverse}
    \end{align}
\end{enumerate}
Then, the following relationship holds,
\begin{equation}
     |{\cal M}| = \sqrt{|{\cal M}^+|^2+|{\cal M}^-|^2}.
     \label{eq:golden_condition}
\end{equation}
\end{prop}
\vspace{0.1mm}
\end{tcolorbox}

The reader familiar with entanglement harvesting will likely recognize these conditions as being satisfied for many cases analyzed in previous literature. Therefore, very often we see this kind of ``collaboration'' between the communication-assisted term and the genuine harvesting term when it comes to detectors getting entangled while they are in causal contact.

For the harvesting setups specified by this proposition, fulfilling either Eq.~\eqref{eq:reverse} or Eq.~\eqref{eq:swapandreverse} makes the setup symmetric under a certain map $\bm T$ in the way described by Corollary \ref{cor:symmetriesT}. When Eq.~\eqref{eq:reverse} holds, the setup becomes symmetric under the following time reversal around $t_{\tc{r}}$,
\begin{equation}
    \bm T(t_{\tc{r}}+t,\bm x, t_{\tc{r}}+t',\bm x')=(t_{\tc{r}}-t,\bm x, t_{\tc{r}}-t',\bm x'),
\end{equation}
where the coordinates are chosen to be the standard Minkowski coordinates. Similarly, when Eq.~\eqref{eq:swapandreverse} holds, the setup becomes symmetric under a different time reversal around $t_{\tc{r}}$ which also includes a time coordinate swap,
\begin{equation}
    \bm T(t_{\tc{r}}+t,\bm x, t_{\tc{r}}+t',\bm x')=(t_{\tc{r}}-t',\bm x, t_{\tc{r}}-t,\bm x').
\end{equation}
The reasons for why the whole setup is symmetric are that, first, Minkowski spacetime and the corresponding vacuum are symmetric under time reversal at all times, allowing $t_{\tc{r}}$ to have an arbitrary value. Second, the conditions over the initial states of the detectors, their trajectories, gaps, and switching functions are sufficient to ensure that the field-detector interactions are time reversal symmetric around some $t_{\tc{r}}$. The details of how Proposition~\ref{prop:switchings} guarantees Corollary \ref{cor:symmetriesT} are provided in Appendices \ref{apx:ExaSym}, \ref{apx:symcoordinates} and \ref{apx:simplify}. Appendix \ref{apx:SymmFermiNormal} shows three more general versions of Proposition~\ref{prop:switchings} for arbitrary harvesting setups that follow the covariant prescription of~\cite{TalesBrunoEdu2020}. 

In summary, the symmetries in Proposition~\ref{prop:switchings} forbid destructive interference. However, it is not difficult to think of physical scenarios where these symmetries are violated. We will illustrate this in the next section, where we will show scenarios with both destructive interference and stronger constructive interference than those described by Proposition~\ref{prop:switchings}.

\section{Entanglement through communication and harvesting can interfere destructively}\label{sec:results}

In this section, we will present concrete examples both for detectors in Minkowski spacetime and in cosmological scenarios where the assumptions of Proposition \ref{prop:switchings} (and their generalizations to cosmological spacetimes) are not satisfied. Concretely, we will analyze scenarios where both the communication-assisted and genuine harvesting contributions to negativity will be greater than the actual negativity acquired by the detectors, showing a destructive interference between the communication contribution to entanglement and the entanglement acquired through harvesting. Moreover, most cases previously studied in the literature satisfy the conditions of Proposition~\ref{prop:switchings} and therefore the two contributions to entanglement while in causal contact were constructively interfering in a particular way, satisfying
\begin{equation}
    |{\cal M}| = \sqrt{|{\cal M}^+|^2+|{\cal M}^-|^2}.
    \label{eq:middleG}
\end{equation}
However, we will see that it is possible to find scenarios where $|{\cal M}|$ is actually larger than~\eqref{eq:middleG} (strong constructive interference) and cases where it is less than~\eqref{eq:middleG} (destructive interference). This showcases that most of the setups previously studied in the literature are in a very particular middle ground between constructive and destructive interference of communication and harvesting.

\subsection{Minkowski spacetime}
\label{sssec:Minkowski_results}
First, let us consider two detectors in 3 + 1 Minkowski spacetime. We adopt coordinates $ x^\mu = (t, \bm x)$ where the metric reads
\begin{equation}
    \dd s^2 = - \dd t^ 2 + \dd \bm{x}^2.
\end{equation}
As usual for two inertial comoving detectors, we will assume that the localization of both detectors A and B is given by spacetime smearing functions which can be split into a switching and spatial smearing in the coordinates $(t, \bm x)$, i.e, \mbox{$\Lambda_{j}(\mf x) = \chi_{j}(t) F_{j}(\bm x)$}. More concretely, we assume that the detectors' centers of mass follow inertial trajectories described by $\bm z_\tc{a}(t) = \bm x_{\tc{a}}$ and $\bm z_\tc{b}(t) = \bm x_{\tc{b}}$, where $\bm x_{\tc{a}}$ and $\bm x_{\tc{b}}$ are constant 3-vectors. We will assume that both detectors start in their respective ground states and that the state of the field is initially prepared in the Minkowski vacuum, the expressions for ${\cal L}_{ij}$ and ${\cal M}$ (see Eqs.~\eqref{eq:LijGen} and~\eqref{eq:MGen}) can be written in the following form (with $\Omega_{\tc{a}} = \Omega_{\tc{b}} = \Omega$):
\begin{align}
& {\cal L}_{ij} = \frac{\lambda^{2}}{2(2\pi)^3}\int{\dd t \dd t'e^{-\ii \Omega(t - t') }\chi_{i}(t) \chi_{j}(t') K_{ij}(t, t')}, \nonumber \\
&{\cal M}  = -\frac{\lambda^2}{2 (2\pi)^3}  \int_{-\infty}^{\infty}{\dd t\int_{-\infty}^{t} \dd t' \  {Q}_{\tc{ab}}(t, t')}.
\label{L_ij_and_M_simplified_Minkowski}
\end{align}
Where we have defined the auxiliary functions
\begin{equation}
        {Q}_{\tc{ab}}(t, t')\! =\! e^{\ii \Omega(t + t')}K_{\tc{ab}}(t, t')  \left(\chi_{\tc{a}}(t) \chi_{\tc{b}}(t')\! +\! \chi_{\tc{a}}(t') \chi_{\tc{b}}(t) \right),
\end{equation}
\begin{equation}
    K_{ij}(t, t')  = 2(2\pi)^3 \int{\dd ^3 \bm x \dd ^3 \bm x' F_{i}(\bm x) F_{j}(\bm x') W(t,\bm x ; t', \bm x')}.
    \label{eq:Kijpre}
\end{equation}
To follow a common choice in the literature, for the spatial smearing function, we take a Gaussian profile, namely
\begin{equation}
    F_{j}(\bm x) = \frac{1}{(\sqrt{\pi} \sigma)^3}e^{-\frac{1}{2}\left(\frac{|\bm x - \bm x_{j}|}{\sigma}\right)^2}.
    \label{shape_gaussian}
\end{equation}
For this particular choice of smearing~\eqref{eq:Kijpre} simplifies to
\begin{equation}
     K_{ij}(t, t') = \int{\frac{\dd \bm k^{3}}{|\bm k|}e^{\ii \bm k \cdot (\bm x_{i} - \bm x_{j})} e^{-|\bm k|^2\sigma^2} e^{-\ii |\bm k|\Delta t }}. 
     \label{eq:Kij}
\end{equation}

Notice that in the equation for ${\cal M}$, the Heaviside theta function that appears in the definition of the Feynman propagator, Eq.~\eqref{eq:defGF}, is implemented as a nested integral in the time variable $t$. To evaluate Eq.~\eqref{eq:Kij}, we will consider separately the cases $i = j$ (same detectors) and $ i \neq j$ (different detectors). First, for $i \neq j$, we have $K_{\tc{ab}}(t, t') = K_{\tc{ba}}(t, t')$, which evaluates to
\begin{equation}
    K_{\tc{ab}}(t, t') = \frac{2 \pi \ii}{d}\left[{\cal I}(\Delta t + d) - {\cal I}(\Delta t - d) \right].
    \label{eq:K_ab_Minkowski}
\end{equation}
Here, $\Delta t = t - t'$, $d = |\bm x_{\tc{a}} - \bm x_{\tc{b}}|$, and
\begin{equation}
    {\cal I}(z) = \frac{\sqrt{\pi}}{2\sigma}e^{-\frac{z^2}{4\sigma^2}} - \frac{\ii}{\sigma} {\cal D}\left(\frac{z}{2 \sigma}\right),
\end{equation}
where ${\cal D}(x) = \frac{\sqrt{\pi}}{2} e^{-x^2}\erfi{x}$ is the Dawson function\footnote{The Dawson function is defined as \begin{equation}
{\cal D}(x)\coloneqq e^{-x^2} \int^x_0 \dd y\, e^{-y^2}.
\end{equation}
}. For $i = j$, we have
\begin{equation}
    K_{jj}(t, t') = \frac{2\pi}{\sigma^2} \left[1 - \frac{\Delta t}{\sqrt{2}\sigma}\left(\ii \sqrt{\frac{\pi}{2}}e^{-\frac{\Delta t^2}{4 \sigma^2}}  + \sqrt{2}{\cal D}\left(\frac{\Delta t}{2 \sigma} \right)\right)\right].
    \label{eq:Kjj_Minkowski}
\end{equation}
Finally, the term ${\cal M}^{-}$ associated with signalling can be written as
\begin{equation}
       {\cal M}^{-}  = -\frac{\lambda^2}{2 (2\pi)^3}  \int_{-\infty}^{\infty}{\dd t\int_{-\infty}^{t} \dd t' \  {Q}^{-}_{\tc{ab}}(t, t')}.
       \label{M_minus_Minkowski}
\end{equation}
Here, the only difference between ${Q}^{-}_{\tc{ab}}(t, t')$ and ${Q}_{\tc{ab}}(t, t')$ is the replacement $K_{\tc{ab}}(t, t') \to K_{\tc{ab}}^{-}(t, t')$, where
\begin{equation}
    K_{\tc{ab}}^{-}(t, t') = \frac{\pi^{3/2} \ii}{\sigma d}\left[e^{-\frac{(\Delta t + d)^2}{4 \sigma^2}} - e^{-\frac{(\Delta t - d)^2}{4 \sigma^2}} \right].
\end{equation}
To proceed with the numerical evaluation of all the relevant terms, we need to define an expression for the switching function of the detectors. Notice that, in the current setup, conditions 1 through 5 of Proposition~\ref{prop:switchings} are satisfied. Therefore, to create a numerical scenario where the conditions of the proposition are not satisfied we need to pick a switching function that violates condition 6.  For this, we can employ switchings that contradict both equation~\eqref{eq:reverse} and equation~\eqref{eq:swapandreverse}. One way to achieve this is by making the switching functions asymmetric. In particular, we are going to choose skew-normal distributions~\cite{SkewSymmetric}, namely
\begin{equation}
    S(t; \alpha) = N(\alpha) e^{-\frac{t^2}{T^2}} \left(1 + \erf\left(\frac{\alpha t}{T} \right) \right).
    \label{eq:skew_normal}
\end{equation}
Here, $N(\alpha)$ stands for a dimensionless normalization factor that ensures that the maximum value of $S(t; \alpha)$ is $1$. Denoting by $t_{\max}(\alpha)$ the point such that \mbox{$S(t_{\max}(\alpha); \alpha) = 1$}, we can define the switchings for \mbox{$j=\text{A}, \text{B}$} as follows:
\begin{equation}
    \chi_{j}(t; \alpha) = S(t - t_{\max}(\alpha) - t_{j}).
    \label{eq:non_symmetric_switching}
\end{equation}
In this way, we ensure that the maximum of the switching happens at $t = t_{j}$ (see figure \ref{fig:non_symmetric_switching}).

\begin{figure}[h!]
    \centering
    \includegraphics[width=8cm]{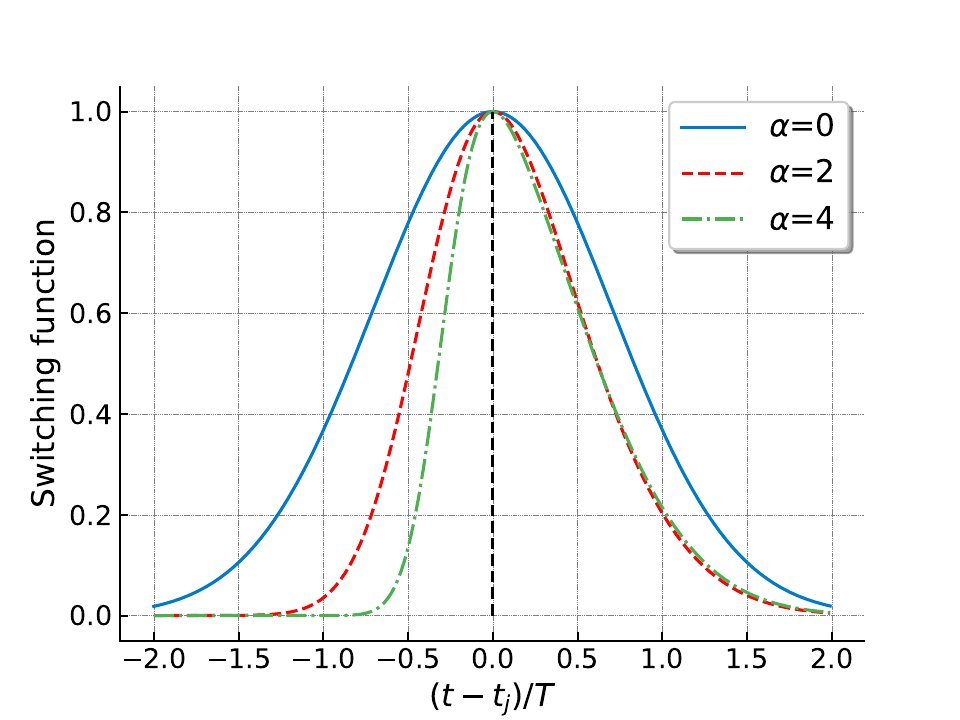}
    \caption{Non-symmetric switching function for different values of the asymmetry parameter $\alpha$.}
    \label{fig:non_symmetric_switching}
\end{figure}

Using the switching of Eq.~\eqref{eq:non_symmetric_switching}, we numerically evaluated the expressions for ${\cal M}$, ${\cal M}^{\pm}$ and ${\cal L}_{jj}$, thus obtaining results for the negativity ${\cal N}$ and the two contributions ${\cal N}^{\pm}$. 

We display in Figure \ref{Non_symmetric_alpha=2.35} different regimes where the conditions in Eqs.~\eqref{eq:reverse} and~\eqref{eq:swapandreverse} are not fulfilled and the interference between ${\cal M}^{+}$ and ${\cal M}^{-}$  affects the detectors' ability to become entangled in different ways. 
 In Fig.~\ref{Non_symmetric_alpha=2.35} we present side-to-side results for the case of asymmetric switching functions ($\alpha = 2.35$, left column) and symmetric Gaussian switching ($\alpha = 0$, right column). Figures~\ref{Non_symmetric_alpha=2.35}a and \ref{Non_symmetric_alpha=2.35}b show $|{\cal M}|$ and $|{\cal M}^{\pm}|$, Figures~\ref{Non_symmetric_alpha=2.35}c and \ref{Non_symmetric_alpha=2.35}d show the relative phases between ${\cal M}^{+}$ and ${\cal M}^{-}$, and Figures~\ref{Non_symmetric_alpha=2.35}e and \ref{Non_symmetric_alpha=2.35}f show ${\cal N}$ and ${\cal N}^{\pm}$. 

Specifically, in the case $\alpha = 2.35$ (see Figures \ref{Non_symmetric_alpha=2.35}a, \ref{Non_symmetric_alpha=2.35}c, \ref{Non_symmetric_alpha=2.35}e), there exists a range of values of $\Delta t = t_{\tc{b}} - t_{\tc{a}}$ that causes the communication-assisted negativity ${\cal N^-}$ to exceed the negativity itself. Moreover, we have identified values of $\Delta t$ where ${\cal N}^{-}$ is non-zero even though ${\cal N} = 0$.

For $\alpha = 0$, condition 6 of Proposition \ref{prop:switchings} is met, so Eq.~\eqref{eq:golden_condition} holds. This means the relative phase between ${\cal M}^{+}$ and ${\cal M}^{-}$ is $\pm \pi/2$, leading to constructive interference. In contrast, for $\alpha=2.35$, there are regions with destructive and constructive interference. Moreover, notice that the constructive interference gets stronger than in the $\pm \pi/2$ relative phase scenario.

Another way of violating condition 6 of Proposition~\eqref{prop:switchings} is to use different shapes for the switching functions of each detector. Indeed, by choosing Gaussian switchings of different widths we can violate the condition in Eq.~\eqref{eq:golden_condition}. Namely, we can choose
\begin{equation}
    \chi_{j}(t) = e^{-\left(\frac{t - t_{j}}{T_{j}} \right)^2},
    \label{gaussian_switching}
\end{equation}
with $T_\tc{a} \neq T_\tc{b}$.

We illustrate the case of switching with two different widths in Figures \ref{plots_grid} and \ref{different_detectors_1D}. We see that there is an asymmetry in how the detectors are entangled depending on which detector switches on first, with an advantage when $\Delta t > 0$ (the detector with shorter switching time (A) couples to the field before the detector with a longer switching time (B)). However, we note that the contributions to negativity, $\mathcal N^+$, and $\cal N^-$, do not depend on the order of coupling. This means that the asymmetry of the negativity with respect to $\Delta t$ is due to interference alone because of how the relative phases change depending on which detector couples first.

Concretely, when detector B couples first  ($\Delta t < 0$) the interference is destructive, reducing the extracted entanglement. Conversely, when A couples first ($\Delta t > 0$) we have constructive interference between the two contributions, with an effect even stronger than that in Eq.~\eqref{eq:middleG}. Hence, in a configuration with a different switching duration in each detector, maximizing entanglement requires carefully timing the detectors' interactions to avoid the adverse effects of destructive interference.
 \begin{figure*}[h]
 \centering
  \includegraphics[width=18cm]{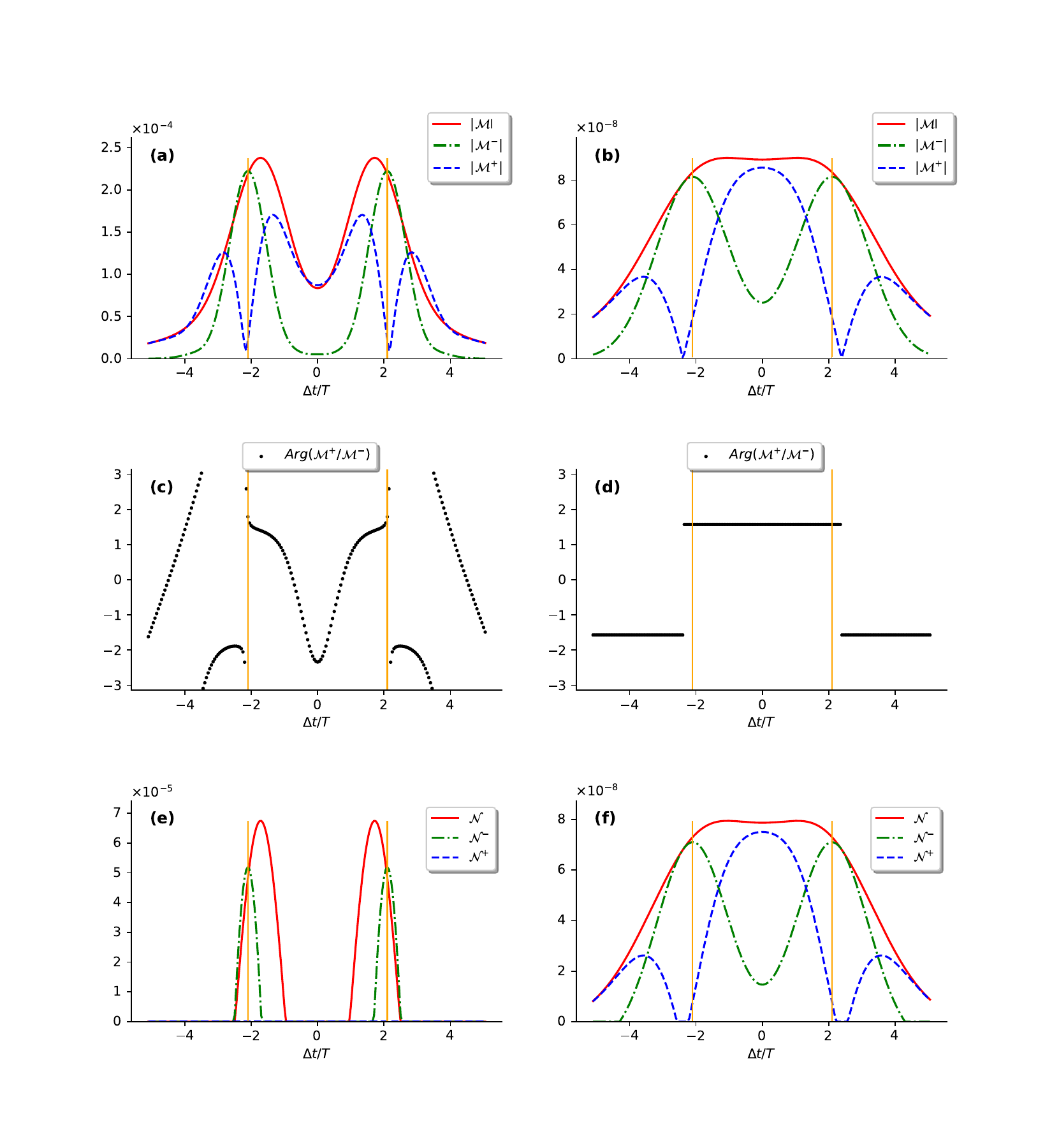}
\caption{Results for an entanglement harvesting setup in Minkowski spacetime with non-symmetric switchings. The physical parameters used are: $\Omega T =5 $, $d/T = 2.1$, and $\sigma/T = 0.3$. In the horizontal axis, we have $\Delta t = t_{\tc{b}} - t_{\tc{a}}$. The solid vertical lines (yellow) represent the lightcones. Figures (a), (c), and (e) correspond to the case $\alpha = 2.35$, whereas Figures (b), (d), and (f) display the case $\alpha = 0$, where we have the standard Gaussian switching. The appearance of an even symmetry with respect to $\Delta t$ in the plots is explained by Appendix \ref{apx:plotSymmetries}.}
\label{Non_symmetric_alpha=2.35}
\end{figure*}

\newpage

\subsection{Example in a cosmological spacetime}\label{sec:deSitter}

By working in a spacetime that has no time-reversal symmetry, we can also construct an entanglement harvesting setup where the interference between the contributions from field correlations and communication violate  condition~\eqref{eq:golden_condition}. As a particular example, we shall consider a spatially flat FRW spacetime, a setting where entanglement harvesting has been studied in the past (see, e.g.,~\cite{collectCalling}). In comoving coordinates $(t, \bm x) = (t, x, y, z)$ the line element reads
\begin{equation}
    \dd s^2 = -\dd t^2 + a(t)^2 (\dd x^2 + \dd y^2 + \dd z^2).
    \label{FRW_comoving_coordinates}
\end{equation}
Or, introducing a new time coordinate $\eta$ 
\begin{equation}
     \eta(t) = \int_{0}^{t}{\dd t' \frac{1}{a(t')}},\label{eq:defConformalTime}
\end{equation}
we explicitly reveal the conformally flat nature of the metric:
\begin{equation}
    \dd s^2 = -a(\eta)^2(\dd t^2 + \dd x^2 + \dd y^2 + \dd z^2).
\end{equation}
We now consider a scalar quantum field $\hat{\phi}$ conformally coupled to curvature so that the field satisfies Eq.~\eqref{KG_eq} with $\xi = \frac{1}{6}$. In this case, it can be shown (see, e.g., \cite{Cosmo}) that the field $\hat{\psi}(\mf x) = a(\eta) \hat{\phi}(\mf x)$ satisfies the Minkowski Klein-Gordon equation. Therefore, the field modes of $\hat{\phi}(\mf x)$ and the Wightman function $W(\mf x, \mf x')$ for this case can be obtained from the ones from the Minkowski case by simply introducing additional multiplicative scale factors $a(\eta)$. 

Now, in order to concretely find setups where the interference between ${\cal M}^{+}$ and ${\cal M}^{-}$ leads to a novel behavior of the negativity, the guiding principle will be the violation of a more general version of Proposition~\ref{prop:switchings} that applies to FRW spacetimes. As demonstrated in Appendix~\ref{apx:proveFRWprop}, we can derive the following proposition:
\begin{tcolorbox}[colframe=blue!20,
colback=blue!5]
\begin{prop}\label{prop:switchingsFRW}
Consider entanglement harvesting setups where
\begin{enumerate}
    \item Spacetime is FRW as in Eq.~\eqref{FRW_comoving_coordinates}.
    \item The scale factor satisfies\\ $a(t+t_{\tc{r}}) = a(t_{\tc{r}}-t)$.
    \item Both detectors are comoving with the Hubble flow and are initially both in the ground state or both in the excited state.
    \item The field is massless, conformally coupled to curvature, and initially in the conformal vacuum.
    \item Spacetime smearings factorize in the comoving frame as the product of a function of only time and a function of only space, i.e., \mbox{$\Lambda_j(t,\bm x)= \mathfrak{T}_j(t)\mathfrak{X}_j(\bm x)$}.
    \item The proper energy gaps are equal, $\Omega_\tc{a}= \Omega_\tc{b}$.
    \item In the comoving frame, for some time $t_{\tc{r}}$, it is satisfied that either
    \begin{align}
        &\mathfrak{T}_j(t+t_{\tc{r}})=\mathfrak{T}_j(t_{\tc{r}}-t),\ j=A,B,\label{eq:reverseFRW}\\
        &\text{or }\; \mathfrak{T}_\tc{a}(t+t_{\tc{r}})=\mathfrak{T}_\tc{b}(t-t_{\tc{r}}).\label{eq:swapandreverseFRW}
    \end{align}
\end{enumerate}
Then, the following relationship holds,
\begin{equation}
     |{\cal M}| = \sqrt{|{\cal M}^+|^2+|{\cal M}^-|^2}.
     \label{eq:golden_conditionFRW}
\end{equation}
\end{prop}
\vspace{0.1mm}
\end{tcolorbox}

Assuming that the initial state of the field is the conformal vacuum and that the detectors both start in their respective ground states, the terms
 ${\cal L}_{ij}$, ${\cal M}$ and ${\cal M}^{-}$ will be given by the same equations~\eqref{L_ij_and_M_simplified_Minkowski} and~\eqref{M_minus_Minkowski} under the replacement
\begin{equation}
    \dd t \dd t' \to \frac{\dd t \dd t'}{a(t) a(t')} \label{eq:replacement}.\\
\end{equation}
In the auxiliary functions $K_{\textsc{ab}}(t, t')$ and $K_{jj}(t, t')$ of equations \ref{eq:K_ab_Minkowski} and~\eqref{eq:Kjj_Minkowski}, we also need to perform the replacement $\Delta t = t - t' \to \Delta \eta = \eta(t) - \eta(t')$. Therefore, the numerical methods applied to compute ${\cal N}$ in the subsection \ref{sssec:Minkowski_results} and all associated terms transition smoothly to this case. All that remains is to specify a particular example of FRW spacetime and choose the shape and switching functions to model our detectors. As a concrete example let us choose the scale factor to be that of De Sitter spacetime:
\begin{equation}
    a(t) = e^{H t}.
    \label{eq:expansion_deSitter}
\end{equation}
\begin{figure}[!b]
\begin{tabular}{c}
  \includegraphics[width=65mm]{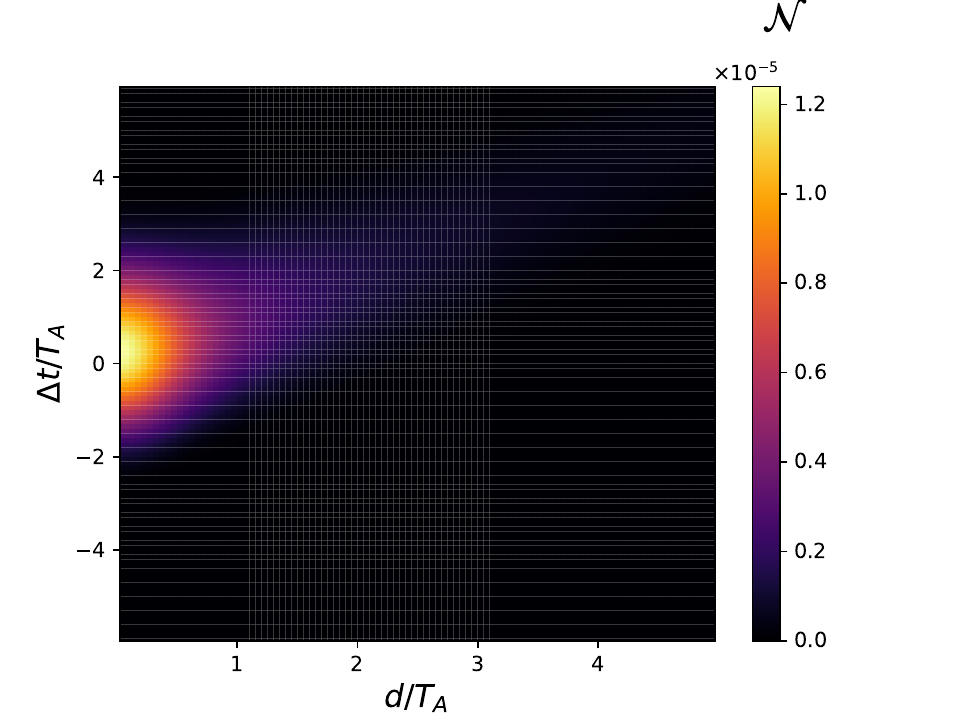} \\   \includegraphics[width=65mm]{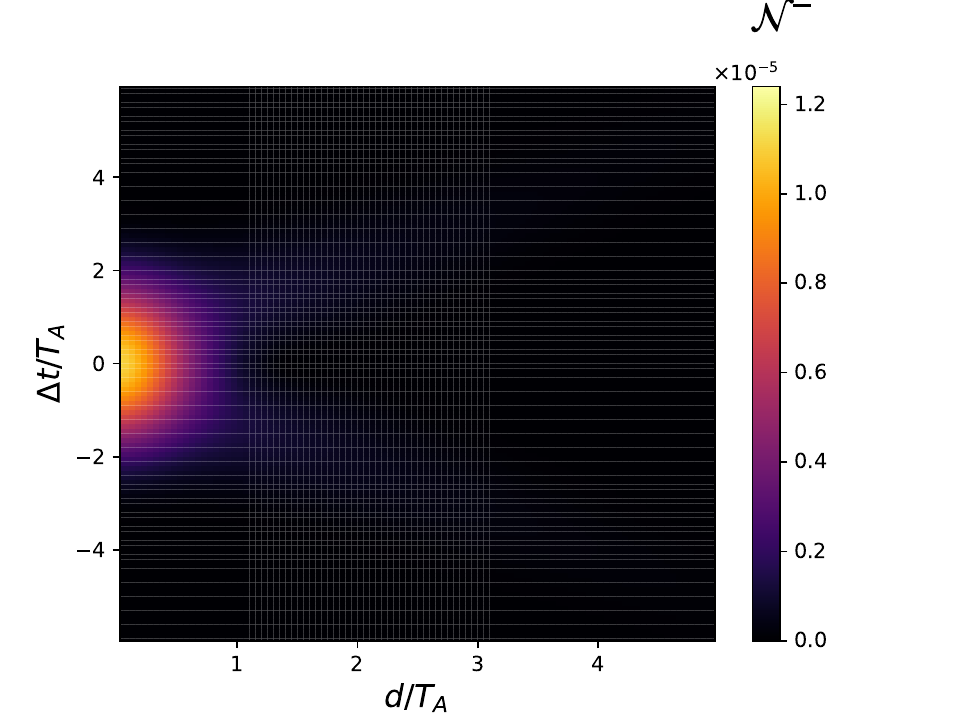}
\end{tabular}
\caption{Numerical results for the negativity (top) and the communication-assisted negativity (bottom) as a function of the proper time $d$ between the centers of the detectors and the time delay $\Delta t$ between the centers of the switchings. The entanglement harvesting setup is in Minkowski spacetime with Gaussian switchings of different length scales, and the physical parameters used are: $\Omega T_{\tc{A}} = 4$, $\sigma/T_{\tc{A}} = 0.2$, and $T_{\tc{B}}/T_{\tc{A}} = 1.3$.
}
  \label{plots_grid}
\end{figure}

\begin{figure}[!h]
  \includegraphics[width=7cm]{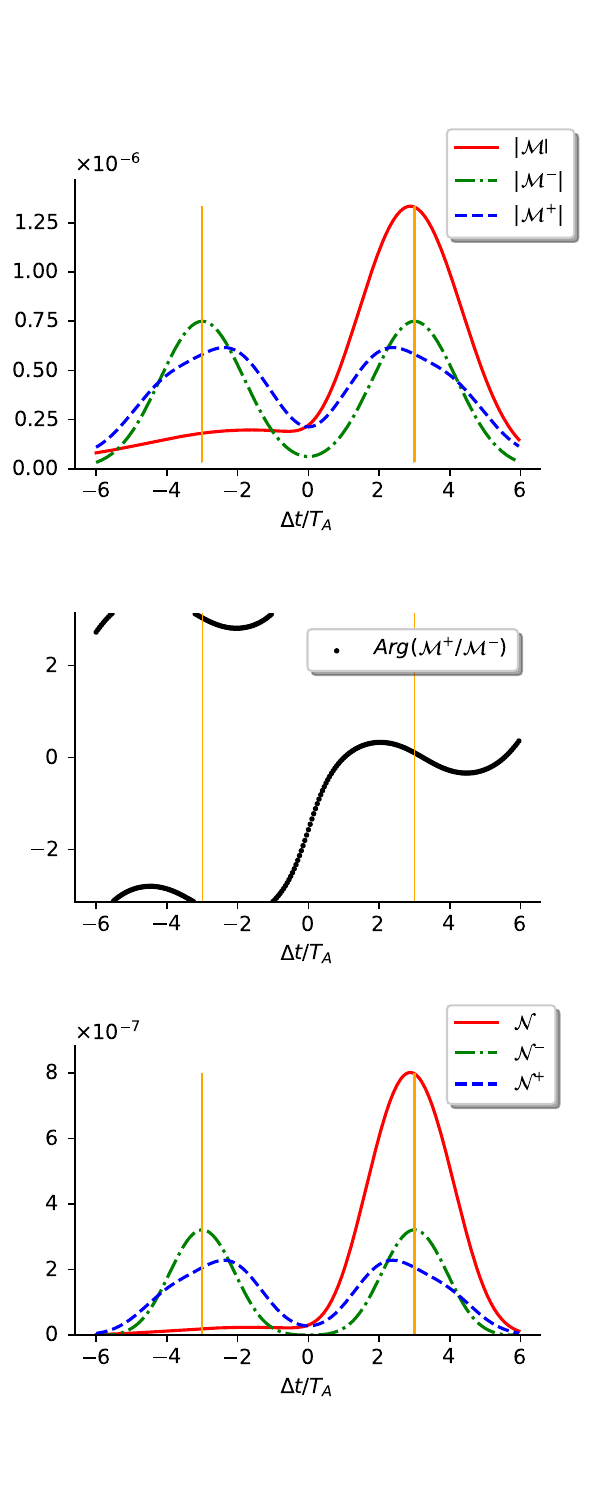}
\caption{Numerical results for negativity, harvested negativity, and communication-assisted negativity as a function of the time delay $\Delta t$ between the center of the switchings. The entanglement harvesting setup is in Minkowski spacetime with Gaussian switchings with different length scales, and the physical parameters used are: $\Omega T_\tc{a} = 4$, $d/T=3$, $\sigma/T_\tc{a} = 0.2$, and $T_\tc{b}/T_\tc{a}=1.3$. The appearance of symmetries with respect to the sign of $\Delta t$ in the plots is explained in Appendix \ref{apx:plotSymmetries}}
\label{different_detectors_1D}
\end{figure}
As for the detectors, let us choose both the switchings and the spatial smearing functions to be Gaussians. That is, the functions $\chi_{j}(t)$ will be given by  
\begin{equation}
    \chi_{j}(t) = e^{-\left(\frac{t - t_{j}}{T} \right)^2}.
    \label{gaussian_switching_FRW}
\end{equation}
The smearing functions $F_{j}(t, \bm x)$ are given by
\begin{equation}
    F_{j}(t, \bm x) = \frac{1}{(\sqrt{2\pi} \sigma a(t))^3}e^{-\frac{1}{2}\left(\frac{|\bm x - \bm x_{j}|}{\sigma}\right)^2}.
\end{equation}
where the detector expands with the expansion of the universe and the scale factor in the denominator comes from the $L^1$ normalization of the smearing function:
\begin{equation}
    \int{\dd \bm x^3 a(t)^3 F_{j}(\bm x, t_{0})} = 1.
\end{equation}
Notice that even if we choose 
\begin{align}
   & \mathfrak{T}_j(t) = a(t)^{-3} \chi_{j}(t), \\
   & \mathfrak{X}_j(\bm x) = a(t)^3 F_{j}(t, \bm x) .
\end{align}
to ensure consistency with assumption $5$ of Proposition~\ref{prop:switchingsFRW}, it is clear that the form of the expansion factor, Eq.~\eqref{eq:expansion_deSitter} automatically leads to the violation of the conditions expressed by Eqs.~\eqref{eq:reverseFRW} and~\eqref{eq:swapandreverseFRW}. Therefore, in the present setup, we have no guarantee that the relation~\eqref{eq:golden_conditionFRW} holds.

A particular case of this kind of setup is displayed in Figure~\ref{deSitter}. Analogously to the instances examined in the Minkowski case, it is observed that the components ${\cal M}^{\pm}$ exhibit interference, leading to scenarios where the communication-mediated negativity surpasses the total negativity. Additionally, around the region $\Delta t/T = 0$, we observe that the harvested negativity ($\mathcal N^+$) can also exceed the value of the negativity ($\mathcal N$). This means that despite the detectors' capability to extract pre-existing correlations from the field, the entanglement acquired through communication destructively interferes with the harvested entanglement, resulting in the overall entanglement being less than each component individually. 


\begin{figure}[b]
  \includegraphics[width=7cm]{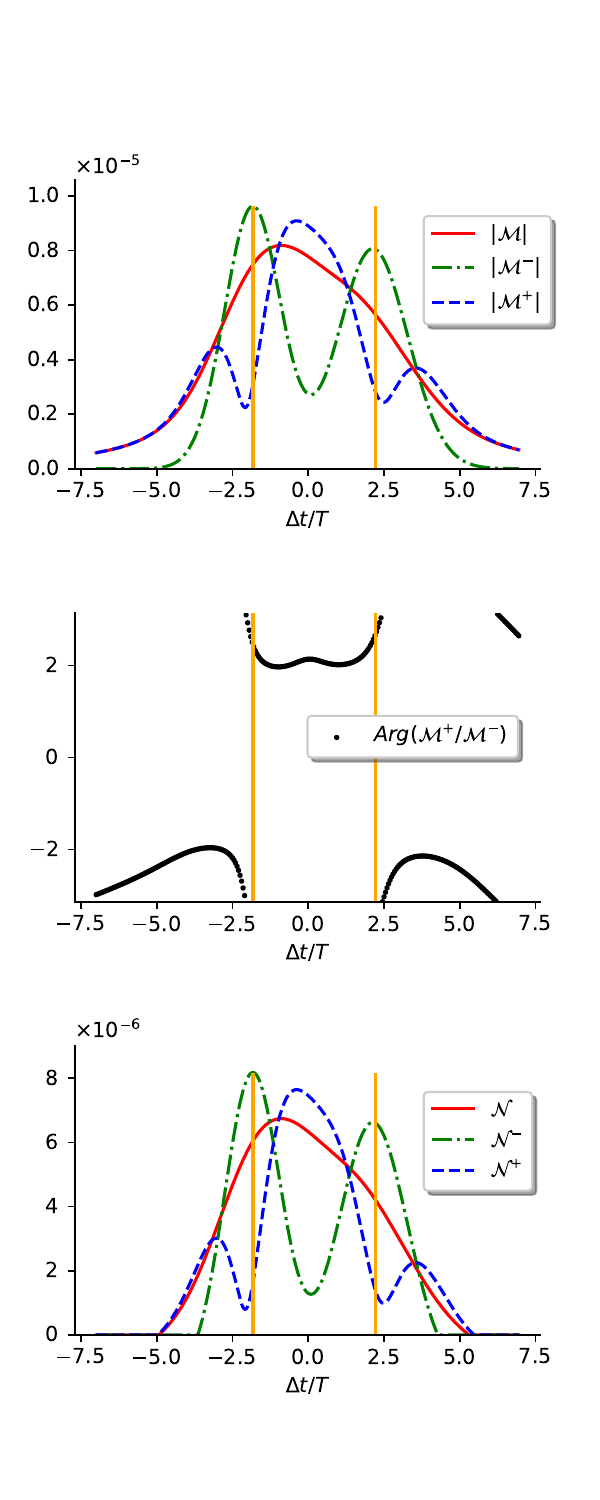}
\caption{Numerical results for De Sitter spacetime. The parameters used are: $\Omega T = 4$, $d/T =2$, $\sigma/T = 0.1$, and $H = 0.1$. Here, $\Delta t = t_{\tc{b}} - t_{\tc{a}}$}
\label{deSitter}
\end{figure}

\section{Conclusions}\label{sec:Conclusions}

We explored the interplay between the two different ways in which two particle detectors can get entangled through their interaction with a quantum field. Namely, 1) harvesting pre-existing entanglement from the field and 2) acquiring entanglement by exchanging information through the field when the detectors are in causal contact. The analysis demonstrated that, in general entanglement harvesting setups, these contributions are not always additive, but instead they can interfere both destructively and constructively. 

We provided sufficient conditions for an entanglement harvesting setup to not display destructive interference between these two contributions. We discuss that the reason this phenomenon was not observed in the past is because these conditions are met in the scenarios commonly explored in the literature, which typically have a lot of symmetry.

We then showed simple examples in Minkowski and cosmological spacetimes using entanglement harvesting setups that violate these conditions and for which the two contributions to harvesting a) combine constructively, amplifying the total amount of entanglement beyond previously observed scenarios and b) interfere destructively, preventing the detectors from getting entangled. For example, whenever the detectors are in causal contact and the switching functions are not time-symmetric, we can see regimes where due to destructive interference the detectors cannot get entangled. This is so even though in these regimes harvesting alone or communication alone would have gotten the detectors entangled.

This study suggests that if one wants to make two detectors entangled through short interactions with a quantum field, one can enhance the amount of entanglement acquired by the detectors by choosing the setup so that the two contributions interfere constructively.

\acknowledgements

The authors would like to thank Tales Rick Perche for helpful discussions. A. T.-B. received the support of a fellowship from ``la Caixa” Foundation (ID 100010434, with fellowship code ). MHZ also thanks Prof. Achim Kempf for funding through his NSERC Discovery grant. Research at Perimeter Institute is supported in part by the Government of Canada through the Department of Innovation, Science and Industry Canada and by the Province of Ontario through the Ministry of Colleges and Universities. EMM acknowledges support through the Discovery Grant Program of the Natural Sciences and Engineering Research Council of Canada (NSERC). EMM  thanks the support from his Ontario Early Researcher award.

\appendix
\section{Inequality between negativity and its communication and harvesting contributions
}
\label{apx:inequality_negativity}

Here, we would like to show that the assumption
\begin{equation}
    |{\cal M}| = \sqrt{|{\cal M}^+|^2+|{\cal M}^-|^2}
    \label{eq:golden_relation_M_appendix}
\end{equation}
implies the following relation between the negativity and its contributions:
\begin{equation}
    {\cal N} \geq \sqrt{({\cal N}^+)^2+({\cal N}^-)^2}.\label{eq:ineqN}
\end{equation}
Let us prove this statement by considering different cases.
\begin{enumerate}[wide, labelwidth=!, labelindent=0pt]
\item {\bf Case 1:} ${\cal N^{+}} = 0$. 

The relation~\eqref{eq:ineqN} becomes equivalent to ${\cal N} \geq {\cal N}^-$, which is directly verified using $|{\cal M}|\geq |{\cal M}^-|$ in the definition of negativity and its associated components, namely Eqs.~\eqref{V_negativity_general} and~\eqref{eq:V_pm}. 

\item {\bf Case 2:} ${\cal N^{-}} = 0$. 

The argument is analogous to case 1.

\item {\bf Case 3}: ${\cal N^{+}} \neq 0$ and ${\cal N^{-}} \neq 0$.

Here, we have ${\cal N}^\pm = {\cal V}^\pm>0$. Let us define
    \begin{equation}
    {\cal L}_\pm = \frac{|{\cal L}_\tc{aa}\pm {\cal L}_\tc{bb}|}{2}.
\end{equation}
In this way, we can write
\begin{align}
     {\cal V}& =\sqrt{|\mathcal{M}|^2 + {\cal L}_-^2}-{\cal L}_+^2, \nonumber \\
    {\cal V}^\pm & =\sqrt{|\mathcal{M}^\pm|^2 + {\cal L}_-^2}-{\cal L}_+^2.
\end{align}
Then, using Eq.~\eqref{eq:golden_relation_M_appendix}, we can cast the statement of Eq.~\eqref{eq:ineqN} into
\begin{align}
    0&\leq {\cal V}^2 - ({\cal V}^+)^2 - ({\cal V}^-)^2
    \\&
    = -{\cal L}_-^2+3{\cal L}_+^2 + 2{\cal L}_+\bigg(\mathcal{V}^+ + \mathcal{V}^- - \sqrt{|\mathcal{M}|^2 + {\cal L}_-^2} \bigg),  \nonumber
\end{align}
Next, we get rid of the square root by moving it to the other side of the inequality and then taking the square on both sides. Further manipulation allows us to rewrite the inequality as
\begin{equation}
    8 {\cal L}_{+}^2 \mathcal{V}^+\mathcal{V}^- + 4 {\cal L}_{+}({\cal L}_{+}^{2}-{\cal L}_{-}^{2})(\mathcal{V}^+ + \mathcal{V}^-) + ({\cal L}_{+}^{2}-{\cal L}_{-}^{2})^2 \geq 0.
\end{equation}
Finally, this statement is true because because of ${\cal L}_+ \geq 0$, ${\cal L}_{+}^{2} \geq {\cal L}_{-}^{2}$ and the assumption $\mathcal{V}^\pm > 0$. Therefore, we conclude that given the relation~\eqref{eq:golden_relation_M_appendix}, the statement~\eqref{eq:ineqN} is true.

\end{enumerate}

\section{Examples of symmetries that satisfy Corollary \ref{cor:symmetriesT}}
\label{apx:ExaSym}

Here, we focus on examples of maps $\bm{T}$ constructed to satisfy the hypotheses of Corollary \ref{cor:symmetriesT}. In particular, we provide sufficient conditions for the harvesting setups so that all the Corollary hypotheses are fulfilled. Moreover, these examples are later used to prove Propositions~\ref{prop:switchings} and \ref{prop:switchingsFRW} of the main text.

\begin{symm} \label{sym:tR} Time reversal around the value  $t_{\tc{r}}$ of a time coordinate in an arbitrary coordinate system $(t,\bm x)$:
    \begin{align}
        &\bm T(\mf x,\mf x')=(T_{\tc{r}}(\mf x),T_{\tc{r}}(\mf x')) ,\nonumber\\
        &T_{\tc{r}}(t_{\tc{r}}+t,\bm x)=(t_{\tc{r}}-t,\bm x).\label{eq:TReverse}
    \end{align}
    The hypotheses of Corollary \ref{cor:symmetriesT} are satisfied if the following assumptions hold:
    \begin{enumerate}
        \item $g(t_{\tc{r}}+t,\bm x)=g(t_{\tc{r}}-t,\bm x)$,
        \item $G_\tc{f}(t_{\tc{r}}+t,\bm x, t_{\tc{r}}+t',\bm x')=G_\tc{f}(t_{\tc{r}}-t,\bm x, t_{\tc{r}}-t',\bm x')$,
        \item $\Lambda_j(t_{\tc{r}}+t,\bm x)=\Lambda_j(t_{\tc{r}}-t,\bm x)$,
        \item For some constant $\nu_j$, 
        $$\mel{\Ostate_j}{\hat\mu_j\big(\tau_j(t_{\tc{r}}+t,\bm x)\big)}{\psi_j}=e^{\ii \nu_j}\mel{\psi_j}{\hat\mu_j\big(\tau_j(t_{\tc{r}}-t,\bm x)\big)}{\Ostate_j}.$$
    \end{enumerate}
    Here, we recall that $g$ is the determinant of the metric, $G_{\tc{f}}$ is the Feynman propagator, $\{\Lambda_{j}\}$ are the spacetime smearing functions, and $\ket{\Ostate_j}$ are the states orthogonal to the initial states $\ket{\psi_j}$ as defined in Eq.~\eqref{eq:states}.
\end{symm}

\begin{symm} \label{sym:ReversalPtSwap} Time reversal around $t_{\tc{r}}$ in the arbitrary coordinates $(t,\bm x)$ \& swap of spacetime points:
    \begin{align}
        &\bm T(\mf x,\mf x')=(T_{\tc{r}}(\mf x'),T_{\tc{r}}(\mf x)),\nonumber\\
        &T_{\tc{r}}(t_{\tc{r}}+t,\bm x)=(t_{\tc{r}}-t,\bm x).\label{eq:TReverseSwap}
    \end{align}
    The hypotheses of Corollary \ref{cor:symmetriesT} are satisfied if the following assumptions hold:
    \begin{enumerate}
        \item $g(t_{\tc{r}}+t,\bm x)=g(t_{\tc{r}}-t,\bm x)$,
        \item $G_\tc{f}(t_{\tc{r}}+t',\bm x', t_{\tc{r}}+t,\bm x)=G_\tc{f}(t_{\tc{r}}-t,\bm x, t_{\tc{r}}-t',\bm x')$,
        \item $\Lambda_\tc{a}(t_{\tc{r}}+t,\bm x)=\Lambda_\tc{b}(t_{\tc{r}}-t,\bm x)$,
        \item For some constant $\nu$, 
        $$\mel{\Ostate_\tc{a}}{\hat\mu_\tc{a}\big(\tau_\tc{a}(t_{\tc{r}}+t,\bm x)\big)}{\psi_\tc{a}}=e^{\ii \nu}\mel{\psi_\tc{b}}{\hat\mu_\tc{b}\big(\tau_\tc{b}(t_{\tc{r}}-t,\bm x)\big)}{\Ostate_\tc{b}}.$$
    \end{enumerate}
\end{symm}

\begin{symm} \label{sym:ReversalPtSwapReflect} Time reversal around $t_{\tc{r}}$ in the arbitrary coordinates $(t,\bm x)$, reflection about $\bm x_0$ in $(t,\bm x)$ \& swap of spacetime points:
    \begin{align}
        &\bm T(\mf x,\mf x')=(T_\tc{rr}(\mf x'),T_\tc{rr}(\mf x)),\nonumber\\
        &T_\tc{rr}(t_{\tc{r}}+t,\bm x_0+\bm x)=(t_{\tc{r}}-t,\bm x_0-\bm x).\label{eq:TReverseSwapReflect}
    \end{align}
    The hypotheses of Corollary \ref{cor:symmetriesT} are satisfied if the following assumptions hold:
    \begin{enumerate}
        \item $g(t_{\tc{r}}+t,\bm x_0+\bm x)=g(t_{\tc{r}}-t,\bm x_0-\bm x)$,
        \item $G_\tc{f}(t_{\tc{r}}+t',\bm x_0+\bm x', t_{\tc{r}}+t,\bm x_0+\bm x)\\=G_\tc{f}(t_{\tc{r}}-t,\bm x_0-\bm x, t_{\tc{r}}-t',\bm x_0-\bm x')$,
        \item $\Lambda_\tc{a}(t_{\tc{r}}+t,\bm x_0+\bm x)=\Lambda_\tc{b}(t_{\tc{r}}-t,\bm x_0-\bm x)$,
        \item For some constant $\nu$,
        \begin{align*}
            &\mel{\Ostate_\tc{a}}{\hat\mu_\tc{a}\big(\tau_\tc{a}(t_{\tc{r}}+t,\bm x_0+\bm x)\big)}{\psi_\tc{a}}\\
            &=e^{\ii \nu}\mel{\psi_\tc{b}}{\hat\mu_\tc{b}\big(\tau_\tc{b}(t_{\tc{r}}-t,\bm x_0-\bm x)\big)}{\Ostate_\tc{b}}.
        \end{align*}
    \end{enumerate}
\end{symm} 
Notice that in this symmetry the distance between the detectors at time $t_{\tc{r}}$ can be adjusted by changing $\bm x_0$, in contrast to Symmetry~\ref{sym:ReversalPtSwap}, where the detectors must be at the same place at time $t_{\tc{r}}$.


\begin{symm} \label{sym:propertimeReversalSwap}    
    Proper time swap and reversal around $\upsilon_\tc{r}/2$.
To better describe this symmetry, we use the coordinates 
$(\tau_j,\widetilde{\bm x}_j)$ where \mbox{$\Lambda_j(\tau_j,\widetilde{\bm x}_j)=\chi_j(\tau_j)F_j(\widetilde{\bm x}_j)$} and where $\tau_j$ are the proper times of the detectors.\footnote{In the usual prescription for detectors, $(\tau_j,\widetilde{\bm x}_j)$ will be chosen to be the Fermi normal coordinates associated to the detectors' trajectories. However, the result given here holds regardless of whether $(\tau_j,\widetilde{\bm x}_j)$ are chosen to be the Fermi normal coordinates.}

Now we can express this symmetry as
    \begin{equation}
        \bm T(\tau_\tc{a},\widetilde{\bm x}_\tc{a}, \tau_\tc{b},\widetilde{\bm x}_\tc{b})=(\upsilon_\tc{r}-\tau_\tc{b},\widetilde{\bm x}_\tc{a},\upsilon_\tc{r}-\tau_\tc{a} ,\widetilde{\bm x}_\tc{b}),
    \end{equation}
    where pairs of spacetime points are expressed in the joint coordinates $(\tau_\tc{a},\widetilde{\bm x}_\tc{a},\tau_\tc{b},\widetilde{\bm x}_\tc{b})$. The hypotheses of Corollary \ref{cor:symmetriesT} are satisfied if the following assumptions hold:
    \begin{enumerate}
        \item $g_\tc{a}(\tau_\tc{a},\widetilde{\bm x}_\tc{a})g_\tc{b}(\tau_\tc{b},\widetilde{\bm x}_\tc{b})
        =g_\tc{a}(\upsilon_\tc{r}-\tau_\tc{b},\widetilde{\bm x}_\tc{a})g_\tc{b}(\upsilon_\tc{r}-\tau_\tc{a},\widetilde{\bm x}_\tc{b})$.
        \item $G_\tc{f}(\tau_\tc{a},\widetilde{\bm x}_\tc{a}, \tau_\tc{b},\widetilde{\bm x}_\tc{b})
        =G_\tc{f}(\upsilon_\tc{r}-\tau_\tc{b},\widetilde{\bm x}_\tc{a}, \upsilon_\tc{r}-\tau_\tc{a},\widetilde{\bm x}_\tc{b})$,
        \item $\chi_\tc{a}(\tau)=\chi_\tc{b}(\upsilon_\tc{r}-\tau)$,
        \item $\Omega_\tc{a} = \Omega_\tc{b}=\Omega$,
        \item $\beta_\tc{a}+\beta_\tc{b}=\Omega\upsilon_\tc{r}$, 
        \item $|\cos \alpha_\tc{a}|=|\cos \alpha_\tc{b}|$, $|\sin \alpha_\tc{a}|=|\sin \alpha_\tc{b}|$.
    \end{enumerate}
    Here, the $g_j$ are the determinant of the metric in $(\tau_j,\widetilde{\bm x}_j)$ coordinates, and the detectors' initial states are
    \begin{equation}
        \ket{\psi_j}=\cos\alpha_j\ket{g_j}+\sin\alpha_j e^{\ii \beta_j}\ket{e_j}.
    \end{equation}
\end{symm} 
Notice that assumptions 4, 5, and 6 appear because they are sufficient to show that
    \begin{align}
        \mel{\Ostate_\tc{a}}{\hat\mu_\tc{a}(\tau)}{\psi_\tc{a}} = e^{\ii \Omega\upsilon_\tc{r}}\mel{\Ostate_\tc{b}}{\hat\mu_\tc{b}(\upsilon_\tc{r}-\tau)}{\psi_\tc{b}},
    \end{align}
    with $\ket{\Ostate_j}$ the states orthogonal to $\ket{\psi_j}$ as defined in Eq.~\eqref{eq:states}.
 
    In order for this $\bm T$ to satisfy Corollary~\ref{cor:symmetriesT}, first notice that $|J(\tau_\tc{a},\widetilde{\bm x}_\tc{a},\tau_\tc{b},\widetilde{\bm x}_\tc{b})|=1$. Then, the symmetry assumptions 1 and 2 directly translate to the conditions required by Eq.~\eqref{eq:symmgGF}. It only remains to show that Eq.~\eqref{eq:conjSymmP} is fulfilled. This condition asks that for some constant angle $\nu$,
    \begin{equation}
        P_\tc{a}^*\big(T_1(\mf x, \mf x')\big)P_\tc{b}^*\big(T_2(\mf x, \mf x')\big) = e^{\ii\nu}P_\tc{a}(\mf x)P_\tc{b}(\mf x').\label{eq:conjSymmP_apx}
    \end{equation}
    Substituting the definition of $P_j$,
    \begin{equation}
         P_j(\mf x) =\Lambda_j(\mf x)\mel{\Ostate_j}{\hat\mu_j(\tau_j(\mf x))}{\psi_j},
    \end{equation}
   and using \mbox{$\Lambda_j(\tau_j,\widetilde{\bm x}_j)=\chi_j(\tau_j)F_j(\widetilde{\bm x}_j)$}, one sees that to satisfy Eq.~\eqref{eq:conjSymmP_apx} the following conditions are enough:
    \begin{align}
        &\chi_\tc{a}(\tau)=\chi_\tc{b}(\upsilon_\tc{r}-\tau),\\
        &\mel{\Ostate_\tc{a}}{\hat\mu_\tc{a}(\tau)}{\psi_\tc{a}}=e^{\ii \nu/2}\mel{\psi_\tc{b}}{\hat\mu_\tc{b}(\upsilon_\tc{r}-\tau)}{\Ostate_\tc{b}}.\label{eq:condMucrossed}
    \end{align}
    The condition over the monopole operators $\hat \mu_j$ simplifies by using the identity
    \begin{align}
        &\mel{\Ostate_j}{\hat\mu_j(\tau_j)}{\psi_j}=e^{\ii\beta_j}\big(c_j e^{\ii(\Omega_j \tau_j -\beta_j)}-s_j e^{-\ii(\Omega_j \tau_j -\beta_j)}\big),
    \end{align}
    where we denote $s_j=\sin^2\alpha_j$, $c_j=\cos^2\alpha_j$. Together with the hypotheses $\Omega_\tc{a}=\Omega_\tc{b}=\Omega$, $\beta_\tc{a}=-\beta_\tc{b}+\Omega\upsilon_\tc{r}$, $c_\tc{a}=c_\tc{b}$, $s_\tc{a}=s_\tc{b}$,
    \begin{align}
        &\mel{\Ostate_\tc{a}}{\hat\mu_\tc{a}(\tau)}{\psi_\tc{a}}\nonumber\\
        &=e^{\ii\beta_\tc{a}}\big(c_\tc{a} e^{\ii(\Omega \tau -\beta_\tc{a})}-s_\tc{a} e^{-\ii(\Omega \tau -\beta_\tc{a})}\big)\nonumber\\
        &=e^{\ii\beta_\tc{a}}\big(c_\tc{b} e^{-\ii(\Omega (\upsilon_\tc{r}-\tau) -\beta_\tc{b})}-s_\tc{b} e^{\ii(\Omega (\upsilon_\tc{r}-\tau) -\beta_\tc{b})}\big)\nonumber\\
        &=e^{\ii \Omega\upsilon_\tc{r}}\mel{\psi_\tc{b}}{\hat\mu_\tc{b}(\upsilon_\tc{r}-\tau)}{\Ostate_\tc{b}}.
    \end{align}
    Therefore, Eq.~\eqref{eq:condMucrossed} is fulfilled with $\nu= 2\Omega\upsilon_\tc{r}$. This completes proving that the assumptions in the statement of Symmetry~\ref{sym:propertimeReversalSwap} are enough to fulfill Corollary \ref{cor:symmetriesT}.

    
\textbf{Summary:} Corollary~\ref{cor:symmetriesT} holds for harvesting setups that meet the conditions described in any of the symmetries presented in this section. We recall that for such symmetric setups, $\cos \Delta \gamma = 0$, i.e.
\begin{equation}
    |{\cal M}| = \sqrt{|{\cal M}^+|^2+|{\cal M}^-|^2},
\end{equation}
which precludes any destructive interference between $\mathcal M^+$ and $\mathcal M^-$. 

As a final remark, a common extension of the UDW detector model includes complex spacetime smearings, i.e. $\Lambda_j^*\neq \Lambda_j$. All the results presented in this section can be generalized to account for complex spacetime smearings, but we have restricted our results to real smearings for simplicity. For complex smearings, the Hamiltonian interaction density for each detector in the interaction picture becomes $\hat h_{I,j}({\mf x})= \lambda \Lambda_j({\mf x})e^{\ii \Omega_{j} \tau_{j}} \hat{\sigma}^{+}_{j} + \text{H.c.}$, with H.c. denoting the Hermitian conjugate. Therefore, the expression for $\mathcal M$ also changes (see, e.g., \cite{carol}). Nonetheless, the requirements specified for each symmetry only need a modification of the conditions of $\Lambda_j$ (or $\chi_j$). Specifically, whenever an equality between smearing or switching functions appears, one of the sides has to be conjugated, and the equality only needs to be fulfilled up to a constant complex phase.

\section{Further simplifications of the hypotheses of Corollary \ref{cor:symmetriesT}}\label{apx:symcoordinates}

In this appendix, we introduce several lemmas that simplify the conditions 3 and 4 of Symmetries \ref{sym:tR}, \ref{sym:ReversalPtSwap} and \ref{sym:ReversalPtSwapReflect} respectively.


\begin{lemma}\label{lemma:LambdaandMu}
    Consider the two sets of coordinates $(\tau_j,\widetilde{\bm x}_j)$ with $j\in\{\text{A},\text{B}\}$, where $\Lambda_j(\tau_j,\widetilde{\bm x}_j)=\chi_j(\tau_j)F_j(\widetilde{\bm x}_j)$. Given coordinates $(t,\bm x)$, the time reversal around $t_{\tc{r}}$ is \mbox{$T_{\tc{r}}(t_{\tc{r}}+t,\bm x)=(t_{\tc{r}}-t,\bm x)$}. 
    Assume that
    \begin{enumerate}
        \item Given $\tau_{\tc{r},j} = \tau_j(t_{\tc{r}}, \bm 0)$, 
        \begin{align}
            &\tau_j(\mf x) = 2 \tau_{\tc{r},j} -\tau_j\big(T_{\tc{r}}(\mf x)\big), \nonumber\\
            &\widetilde{\bm x}_j(\mf x) = \widetilde{\bm x}_j\big(T_{\tc{r}}(\mf x)\big).\label{eq:hypLemmaLmu1}
        \end{align} 
        \item $\chi_j(\tau)=\chi_j(2\tau_{\tc{r},j}-\tau)$.
        \item $\beta_j = \Omega_j\tau_{\tc{r},j}$.
    \end{enumerate} 
    Recall that $\beta_j$ is the relative phase in the detectors' initial states $\ket{\psi_j}=\cos\alpha_j\ket{g_j}+\sin\alpha_j e^{\ii \beta_j}\ket{e_j}$.
    
    Under these assumptions, the conditions 3 and 4 of Symmetry~\ref{sym:tR} are fulfilled.
    
\end{lemma}
\begin{proof} 
    We prove the condition 3 first, to do so, keep in mind that
    \begin{equation}
        \Lambda_j(t,\bm x) = \chi_j\big(\tau_j(t,\bm x)\big)F_j\big(\widetilde{\bm x}_j(t,\bm x)\big).
    \end{equation}
    Then, using the time reversibility condition over $\chi_j$,
    \begin{align}
        \chi_j\big(\tau_j(t_{\tc{r}}-t,\bm x)\big) 
        & = \chi_j\big(2\tau_{\tc{r},j}-\tau_j(t_{\tc{r}}-t,\bm x)\big) \nonumber \\
        & = \chi_j\big(\tau_j(t_{\tc{r}}+t,\bm x)\big).
    \end{align}
    Substituting back into the first equation shows condition 3 of Symmetry~\ref{sym:tR},
    \begin{align}
        \Lambda_j(t_{\tc{r}}-t,\bm x) &= \chi_j\big(\tau_j(t_{\tc{r}}-t,\bm x)\big)F_j\big(\widetilde{\bm x}_j(t_{\tc{r}}-t,\bm x)\big)\nonumber\\
        &=\chi_j\big(\tau_j(t_{\tc{r}}+t,\bm x)\big)F_j\big(\widetilde{\bm x}_j(t_{\tc{r}}+t,\bm x)\big)\nonumber\\
        &= \Lambda_j(t_{\tc{r}}+t,\bm x) .
    \end{align}
    
    To prove the condition 4 of Symmetry~\ref{sym:tR}, let us use the identity
    \begin{align}
        &\mel{\Ostate_j}{\hat\mu_j(\tau_j)}{\psi_j}=e^{\ii\beta_j}\big(c_j e^{\ii(\Omega_j \tau_j -\beta_j)}-s_j e^{-\ii(\Omega_j \tau_j -\beta_j)}\big),
    \end{align}
    where $s_j=\sin^2\alpha_j$, $c_j=\cos^2\alpha_j$ and we recall that $\ket{\Ostate_j}$ are the states orthogonal to the initial state $\ket{\psi_j}$ as defined in Eq.~\eqref{eq:states}. Substituting into the following expression,
    \begin{align}\label{eq:mureflection}
        &\mel{\Ostate_j}{\hat\mu_j\big(\tau_j(t_{\tc{r}}-t,\bm x)\big)}{\psi_j}\nonumber\\
        &=e^{\ii\Omega_j\tau_{\tc{r},j}}\Big(c_j e^{\ii\big(\Omega_j \tau_j(t_{\tc{r}}-t,\bm x) -\Omega_j\tau_{\tc{r},j}\big)}\nonumber\\
        &\qquad\qquad\quad -s_j e^{-\ii\big(\Omega_j \tau_j(t_{\tc{r}}-t,\bm x) -\Omega_j\tau_{\tc{r},j}\big)}\Big)\nonumber\\
        &=e^{\ii\Omega_j\tau_{\tc{r},j}}\Big(c_j e^{-\ii\big(\Omega_j \tau_j(t_{\tc{r}}+t,\bm x) -\Omega_j\tau_{\tc{r},j}\big)}\nonumber\\
        &\qquad\qquad\quad -s_j e^{\ii\big(\Omega_j \tau_j(t_{\tc{r}}+t,\bm x) -\Omega_j\tau_{\tc{r},j}\big)}\Big)\nonumber\\
        &=e^{2\ii\Omega_j\tau_{\tc{r},j}}\mel{\psi_j}{\hat\mu_j\big(\tau_j(t_{\tc{r}}+t,\bm x)\big)}{\Ostate_j},
    \end{align}
    where the condition over $\tau_j$ given by Eq.~\eqref{eq:hypLemmaLmu1}, and that \mbox{$\beta_j = \Omega_j\tau_{\tc{r},j}$}, finishing the proof.
\end{proof}

\begin{lemma}\label{lemma:LambdaandMuSwap}
     Consider the two sets of coordinates $(\tau_j,\widetilde{\bm x}_j)$ with $j\in\{\text{A},\text{B}\}$, where $\Lambda_j(\tau_j,\widetilde{\bm x}_j)=\chi_j(\tau_j)F_j(\widetilde{\bm x}_j)$. Given coordinates $(t,\bm x)$, the time reversal around $t_{\tc{r}}$ is \mbox{$T_{\tc{r}}(t_{\tc{r}}+t,\bm x)=(t_{\tc{r}}-t,\bm x)$}. Assume that
     \begin{enumerate}
         \item Given $\upsilon_\tc{r} = \tau_\tc{a}(t_{\tc{r}}, \bm 0)+\tau_\tc{b}(t_{\tc{r}}, \bm 0)$, 
         \begin{align}
            &\tau_\tc{a}(\mf x) = \upsilon_\tc{r} -\tau_\tc{b}\big(T_{\tc{r}}(\mf x)\big), \nonumber\\
            &\widetilde{\bm x}_\tc{a}(\mf x) = \widetilde{\bm x}_\tc{b}\big(T_{\tc{r}}(\mf x)\big).\label{eq:hypLemmaLmu2}
        \end{align}
        \item $F_\tc{a}(\widetilde{\bm x})=F_\tc{b}(\widetilde{\bm x})$.
        \item $\chi_\tc{a}(\tau)=\chi_\tc{b}(\upsilon_\tc{r}-\tau)$.
        \item $\Omega_\tc{a} = \Omega_\tc{b}=\Omega$.
        \item $\beta_\tc{a} + \beta_\tc{b}=\Omega\upsilon_\tc{r}$. 
        \item $|\cos \alpha_\tc{a}|=|\cos \alpha_\tc{b}|$, $|\sin \alpha_\tc{a}|=|\sin \alpha_\tc{b}|$, 
    \end{enumerate}
    where we recall that the detectors' initial states are \mbox{$\ket{\psi_j}=\cos\alpha_j\ket{g_j}+\sin\alpha_j e^{\ii \beta_j}\ket{e_j}$}.
    
    Under these assumptions, conditions 3 and 4 of Symmetry~\ref{sym:ReversalPtSwap} are fulfilled.
\end{lemma}
\begin{proof} 
    We prove condition 3 first, to do so, keep in mind that
    \begin{equation}
        \Lambda_j(t,\bm x) = \chi_j\big(\tau_j(t,\bm x)\big)F_j\big(\widetilde{\bm x}_j(t,\bm x)\big).
    \end{equation}
    Then, using the time reversibility condition over the $\chi_j$,
    \begin{align}
        \chi_\tc{a}\big(\tau_\tc{a}(t_{\tc{r}}-t,\bm x)\big) 
        & = \chi_\tc{b}\big(\upsilon_\tc{r}-\tau_\tc{a}(t_{\tc{r}}-t,\bm x)\big) \nonumber \\
        & = \chi_\tc{b}\big(\tau_\tc{b}(t_{\tc{r}}+t,\bm x)\big),
    \end{align}
    and similarly,
    \begin{equation}
        F_\tc{a}\big(\widetilde{\bm x}_\tc{a}(t_{\tc{r}}-t, \bm x)\big) = F_\tc{b}\big(\widetilde{\bm x}_\tc{b}(t_{\tc{r}}+t, \bm x)\big).
    \end{equation}
    Substituting these back into the first equation completes the first half of the proof,
    \begin{align}
        \Lambda_\tc{a}(t_{\tc{r}}-t,\bm x) &= \chi_\tc{a}\big(\tau_\tc{a}(t_{\tc{r}}-t,\bm x)\big)F_\tc{a}\big(\widetilde{\bm x}_\tc{a}(t_{\tc{r}}-t,\bm x)\big)\nonumber\\
        &=\chi_\tc{b}\big(\tau_\tc{b}(t_{\tc{r}}+t,\bm x)\big)F_\tc{b}\big(\widetilde{\bm x}_j(t_{\tc{r}}+t,\bm x)\big)\nonumber\\
        &= \Lambda_\tc{b}(t_{\tc{r}}+t,\bm x) .
    \end{align}
    
    To prove the condition 4 of Symmetry~\ref{sym:ReversalPtSwap}, let us use the identity
    \begin{align}
        &\mel{\Ostate_j}{\hat\mu_j(\tau_j)}{\psi_j}=e^{\ii\beta_j}\big(c_j e^{\ii(\Omega_j \tau_j -\beta_j)}-s_j e^{-\ii(\Omega_j \tau_j -\beta_j)}\big),
    \end{align}
    where $s_j=\sin^2\alpha_j$, $c_j=\cos^2\alpha_j$ and we recall that $\ket{\Ostate_j}$ are the states orthogonal to the initial state $\ket{\psi_j}$ as defined in Eq.~\eqref{eq:states}. The proof finishes by combining this identity with the hypotheses \mbox{$c_\tc{a}=c_\tc{b}$}, \mbox{$s_\tc{a}=s_\tc{b}$}, $\Omega_\tc{a}=\Omega_\tc{b}=\Omega$, \mbox{$\beta_\tc{a}=-\beta_\tc{b}+\Omega \upsilon_\tc{r}$}, and with the assumption over the $\tau_j$,
    \begin{align}
        &\mel{\Ostate_\tc{a}}{\hat\mu_\tc{a}\big(\tau_\tc{a}(t_{\tc{r}}-t,\bm x)\big)}{\psi_\tc{a}}\nonumber\\
        &=e^{\ii\beta_\tc{a}}\big(c_\tc{a} e^{\ii(\Omega \tau_\tc{a}(t_{\tc{r}}-t,\bm x) -\beta_\tc{a})}-s_\tc{a} e^{-\ii(\Omega \tau_\tc{a}(t_{\tc{r}}-t,\bm x)-\beta_\tc{a})}\big)\nonumber\\
        &=e^{\ii\beta_\tc{a}}\big(c_\tc{b} e^{-\ii(\Omega \tau_\tc{b}(t_{\tc{r}}+t,\bm x) -\beta_\tc{b})}-s_\tc{b} e^{\ii(\Omega \tau_\tc{b}(t_{\tc{r}}+t,\bm x) -\beta_\tc{b})}\big)\nonumber\\
        &=e^{\ii \Omega \upsilon_\tc{r}}\mel{\psi_\tc{b}}{\hat\mu_\tc{b}\big(\tau_\tc{b}(t_{\tc{r}}+t,\bm x)\big)}{\Ostate_\tc{b}}.
    \end{align}
\end{proof}

\begin{lemma}\label{lemma:LambdaandMuSwapReflect}
    Consider the two sets of coordinates $(\tau_j,\widetilde{\bm x}_j)$ with $j\in\{\text{A},\text{B}\}$, where $\Lambda_j(\tau_j,\widetilde{\bm x}_j)=\chi_j(\tau_j)F_j(\widetilde{\bm x}_j)$. Given coordinates $(t,\bm x)$, the time reversal around $t_{\tc{r}}$ and reflection around $\bm x_0$ is $T_{\tc{rr}}(t,\bm x)=(2t_{\tc{r}}-t,2\bm x_0-\bm x)$. Assume that
     \begin{enumerate}
         \item Given $\upsilon_\tc{r} = \tau_\tc{a}(t_{\tc{r}}, \bm 0)+\tau_\tc{b}(t_{\tc{r}}, \bm 0)$, 
         \begin{align}
            &\tau_\tc{a}(\mf x) = \upsilon_\tc{r} -\tau_\tc{b}\big(T_{\tc{rr}}(\mf x)\big), \nonumber\\
            &\widetilde{\bm x}_\tc{a}(\mf x) = \widetilde{\bm x}_\tc{b}\big(T_{\tc{rr}}(\mf x)\big).\label{eq:hypLemmaLmu3}
        \end{align}
        \item $F_\tc{a}(\widetilde{\bm x})=F_\tc{b}(\widetilde{\bm x})$.
        \item $\chi_\tc{a}(\tau)=\chi_\tc{b}(\upsilon_\tc{r}-\tau)$.
        \item $\Omega_\tc{a} = \Omega_\tc{b}=\Omega$.
        \item $\beta_\tc{a} + \beta_\tc{b}=\Omega\upsilon_\tc{r}$. 
        \item $|\cos \alpha_\tc{a}|=|\cos \alpha_\tc{b}|$, $|\sin \alpha_\tc{a}|=|\sin \alpha_\tc{b}|$.
    \end{enumerate}
    where we recall that the detectors' initial states are \mbox{$\ket{\psi_j}=\cos\alpha_j\ket{g_j}+\sin\alpha_j e^{\ii \beta_j}\ket{e_j}$}.
    
    Under these assumptions, the conditions 3 and 4 of Symmetry~\ref{sym:ReversalPtSwapReflect} are fulfilled.
\end{lemma}
This lemma is proven analogously to the previous Lemma~\ref{lemma:LambdaandMuSwap}. Notice the following apparent incompatibility: the lemma asks the smearing functions $F_j$ to be equal, while the assumptions of Symmetry~\ref{sym:ReversalPtSwapReflect} ask spacetime smearings $\Lambda_j$ to be related by a spatial reflection and a time reversal. However, the equality between $F_j$ is actually expected because the spatial smearing functions $F_j$ are defined in their corresponding prescription coordinates $(\tau_j,\widetilde{\bm x}_j)$, which we assumed to be related by a time reflection and a space reflection, as in Eq.~\eqref{eq:hypLemmaLmu3}.

\section{Simplifying the symmetries for the most common  entanglement harvesting scenarios}\label{apx:simplify}

Many of the scenarios analyzed in the literature of entanglement harvesting already restrict themselves to very symmetric particular setups. While this is done for different reasons (mainly simplifying the math) this yields important simplifications of the conditions required for each of the $\bm T$ maps described in Appendix~\ref{apx:ExaSym} to be symmetries that fulfill the assumptions of Corollary \ref{cor:symmetriesT}. 

In particular,  we will show how the usual assumptions in the literature for harvesting in flat spacetime and FRW yield Propositions~\ref{prop:switchings} and \ref{prop:switchingsFRW} of the main text. 

\subsection{Sufficient symmetries for harvesting and communication collaboration in Minkowski spacetime}
This Subappendix proves Proposition~\ref{prop:switchings} of the main text. This proof consists of showing that the conditions for the symmetries provided in Appendix~\ref{apx:ExaSym} are fulfilled in Minkowski spacetime, for fields that are prepared in the vacuum and inertial, comoving, symmetrically switched detectors that are prepared both in the ground state or both in the excited state. 


\begin{lemma}\label{lemma:MinkowskiGF} Assume that 
\begin{enumerate}
    \item The spacetime is flat.
    \item The field starts in the Minkowski vacuum.
\end{enumerate}
This implies that f any inertial coordinate system $(t,\bm x)$ and any constants $t_{\tc{r}}$ and $\bm x_0$, the Feynmann propagator $G_\tc{f}$ fulfills 
\begin{align}
     G_\tc{f}(t,\bm x, t',\bm x')&=G_\tc{f}(2t_{\tc{r}}-t',2\bm x_0-\bm x', 2t_{\tc{r}}-t,2\bm x_0-\bm x)\nonumber\\
     &=G_\tc{f}(2t_{\tc{r}}-t',\bm x', 2t_{\tc{r}}-t,\bm x)\nonumber\\
     &=G_\tc{f}(2t_{\tc{r}}-t,\bm x,2t_{\tc{r}}-t',\bm x').\label{eq:symGF}
\end{align}
And therefore, these relations fulfill the condition 2 for the Symmetries~\ref{sym:tR}, \ref{sym:ReversalPtSwap}, \ref{sym:ReversalPtSwapReflect}. Moreover, since $g(t,\bm x)=-1$, the condition 1 of all these symmetries is also fulfilled.
\end{lemma} 
\begin{proof}
    The Wightman of the Minkowski vacuum for a scalar field is
    \begin{align}
        W(t,\bm x,t',\bm x')  & =  \frac{1}{(2\pi)^n}\int{ \ \frac{\dd ^{n} \bm k}{2 \omega_{\bm k}} e^{\ii \bm k \cdot(\bm x - \bm x')}e^{-\ii \omega_{\bm k}(t - t')}},
    \end{align}
    with $\omega_{\bm k} = \sqrt{\bm k^2 + m^2}$.
    The Minkowski vacuum Wightman function fulfills 
    \begin{equation}
        W(t,\bm x, t', \bm x')=W(\Delta t,|\Delta \bm x|),
    \end{equation}
    where we denoted $W(\Delta t,|\Delta \bm x|)= W(\Delta t,|\Delta \bm x|,0,\bm 0)$ and $\Delta t = t- t',\ \Delta \bm x = \bm x-\bm x'$. Then, from the definition of $G_\tc{f}$ in Eq.~\eqref{eq:defGF}, we have
    \begin{align}
       &G_\tc{f}(t,\bm x, t', \bm x') \nonumber\\
       &=\Theta(\Delta t) W(\Delta t,|\Delta \bm x|) + 
        \Theta(-\Delta t) W(-\Delta t,|\Delta \bm x|)\nonumber\\
        &=W(|\Delta t|,|\Delta \bm x|).
    \end{align}
    Therefore, Eq.~\eqref{eq:symGF} is fulfilled because all of the equated expressions for $G_\tc{f}$ have the same values of $|\Delta t|$ and $|\Delta \bm x|$.    
\end{proof}

\begin{lemma} \label{lemma:sym1Cond34} Assume that, for a given constant $t_\tc{r}$,
    \begin{enumerate}
        \item The spacetime is flat.
        \item The field starts in the Minkowski vacuum.
        \item The detectors start both in their ground state or both in their excited state. 
        \item The detectors are inertial and comoving. 
        \item In the comoving frame $(t,\bm x)$, the spacetime smearing factorizes $\Lambda_j(t,\bm x)=\chi_j(t)F_j(\bm x)$.
        \item $\chi_j(t)=\chi_j(2t_{\tc{r}}-t)$.
    \end{enumerate}
    Then, the conditions of Symmetry~\ref{sym:tR} are fulfilled.
\end{lemma}
\begin{proof}
    First, conditions 1 and 2 of Symmetry~\ref{sym:tR} hold because of Lemma \ref{lemma:MinkowskiGF}. 

    To prove the conditions 3 and 4 of Symmetry~\ref{sym:tR}, we show next that the assumptions of Lemma~\ref{lemma:LambdaandMu} hold. First, choose both coordinate systems $(\tau_\tc{a},\widetilde{\bm x}_\tc{a})$ and $(\tau_\tc{b},\widetilde{\bm x}_\tc{b})$ to be $(t,\bm x)$. Then, the assumption 1 of Lemma~\ref{lemma:LambdaandMu} holds. The assumption 2 holds because of $\chi_j(t)=\chi_j(2t_{\tc{r}}-t)$ and $t_\tc{r} = \tau_{\tc{r},j}$. Finally, the assumption 3 upon the relative phases $\beta_j$ is trivially satisfied because both initial states are either $\ket{\psi_j} = \ket{g_j}$ or $\ket{\psi_j} =\ket{e_j}$.

\end{proof}
\begin{lemma} \label{lemma:sym2Cond34} Assume that, for a given constant $t_\tc{r}$,
    \begin{enumerate}
        \item The spacetime is flat.
        \item The field starts in the Minkowski vacuum.
        \item The detectors start both in their ground state or both in their excited state. 
        \item The detectors are inertial and comoving. 
        \item In the comoving frame $(t,\bm x)$, the spacetime smearing factorizes $\Lambda_j(t,\bm x)=\chi_j(t)F_j(\bm x)$.
        \item $\Omega_\tc{a} = \Omega_\tc{b}$.
        \item $\chi_\tc{a}(t)=\chi_\tc{b}(2t_{\tc{r}}-t)$.
        \item $F_\tc{a}(\bm x)=F_\tc{b}(\bm x)$.
    \end{enumerate}
    Then, the conditions of Symmetry~\ref{sym:ReversalPtSwap} are fulfilled.
\end{lemma}
\begin{proof}
    Shown analogously to Lemma~\ref{lemma:sym1Cond34}, using Lemma~\ref{lemma:LambdaandMuSwap} instead of Lemma~\ref{lemma:LambdaandMu}. Additionally, \mbox{$|\cos\alpha_\tc{a}|=|\cos\alpha_\tc{b}|$} and \mbox{$|\sin\alpha_\tc{a}|=|\sin\alpha_\tc{b}|$} hold because of restricting the initial states of the detectors to be both ground or both excited.
\end{proof}

\begin{lemma} \label{lemma:sym3Cond34} Assume that, for a given constant $t_\tc{r}$,
    \begin{enumerate}
        \item The spacetime is flat.
        \item The field starts in the Minkowski vacuum.
        \item The detectors start both in their ground state or both in their excited state. 
        \item The detectors are inertial and comoving. 
        \item In the comoving frame $(t,\bm x)$, the spacetime smearing factorizes $\Lambda_j(t,\bm x)=\chi_j(t)F_j(\bm x)$.
        \item $\Omega_\tc{a} = \Omega_\tc{b}$.
        \item $\chi_\tc{a}(t)=\chi_\tc{b}(2t_{\tc{r}}-t)$.
        \item $F_\tc{a}(\bm x)=F_\tc{b}(2\bm x_0 - \bm x)$.
    \end{enumerate}
    Then, the conditions of Symmetry~\ref{sym:ReversalPtSwapReflect} are fulfilled.
\end{lemma}
\begin{proof}
    Shown analogously to Lemma~\ref{lemma:sym1Cond34}, using Lemma~\ref{lemma:LambdaandMuSwapReflect} instead of Lemma~\ref{lemma:LambdaandMu}. The major difference is that while we still choose the coordinate system $(\tau_\tc{a},\widetilde{\bm x}_\tc{a})$ to be $(t,\bm x)$, here we choose $(\tau_\tc{b},\widetilde{\bm x}_\tc{b})$ to be instead
    \begin{equation}
        \tau_\tc{b}(t,\bm x) = t,\ \widetilde{\bm x}_\tc{b}(t,\bm x) = 2\bm x_0-\bm x.
    \end{equation}
\end{proof}


\begin{lemma} \label{lemma:sym4_Simple} Assume that, for a given constant $t_\tc{r}$,
\begin{enumerate}
    \item The spacetime is flat.
    \item The field starts in the Minkowski vacuum.
    \item The detectors start both in their ground state or both in their excited state. 
    \item The detectors are inertial and comoving. 
    \item In the comoving frame $(t,\bm x)$, the spacetime smearing factorizes $\Lambda_j(t,\bm x)=\chi_j(t)F_j(\bm x)$.
    \item $\Omega_\tc{a} = \Omega_\tc{b}$.
    \item $\chi_\tc{a}(t)=\chi_\tc{b}(2t_{\tc{r}}-t)$.
\end{enumerate}
Then, the conditions of Symmetry~\ref{sym:propertimeReversalSwap} are fulfilled, by picking both coordinate systems $(\tau_j,\widetilde{\bm x}_j)$ to be $(t,\bm x)$.
\end{lemma}
\begin{proof}
    Since we picked both $(\tau_j,\widetilde{\bm x}_j)$ to be the inertial comoving frame $(t,\bm x)$, then the determinants of the metric in coordinates  $(\tau_j,\widetilde{\bm x}_j)$ satisfy \mbox{$g_\tc{a} = g_\tc{b} = -1$} under the given assumptions, hence fulfilling the condition 1 of Symmetry~\ref{sym:propertimeReversalSwap}. The condition 2 over $G_\tc{f}$ becomes the following: for a given constant $t_\tc{r}=\upsilon_\tc{r}/2$,
    \begin{align}
        G_\tc{f}(t,\bm x, t',\bm x')=G_\tc{f}(2t_{\tc{r}}-t,\bm x,2t_{\tc{r}}-t',\bm x'),
    \end{align}
    which is true regardless of $t_\tc{r}$ for the Minkowski vacuum, as shown in Lemma~\ref{lemma:MinkowskiGF}.

    The conditions 3 and 4 of Symmetry~\ref{sym:propertimeReversalSwap} are the assumptions 6 and 7 of the current lemma, using the relationship $t_\tc{r}=\upsilon_\tc{r}/2$.

    Finally, condition 5 upon the relative phases $\beta_j$ is trivially satisfied because of the initial states of the detectors are either both the ground or both the excited state, and condition 6 (\mbox{$|\cos\alpha_\tc{a}|=|\cos\alpha_\tc{b}|$}, \mbox{$|\sin\alpha_\tc{a}|=|\sin\alpha_\tc{b}|$}) is fulfilled as well for the same reason.
\end{proof}

Finally,  Proposition~\ref{prop:switchings} of the main text follows from combining the assumptions of the Lemmas~\ref{lemma:sym1Cond34}, \ref{lemma:sym2Cond34}, \ref{lemma:sym3Cond34}, \ref{lemma:sym4_Simple}. 
Specifically, fulfilling any of these lemmas is enough to have the corresponding symmetry and thus for Corollary \ref{cor:symmetriesT} to hold. It is worth noticing that although we included one lemma per symmetry of Appendix~\ref{apx:ExaSym} for completion,  Lemmas~\ref{lemma:sym2Cond34} and \ref{lemma:sym3Cond34} become redundant after Lemma~\ref{lemma:sym4_Simple}, because whenever Lemmas~\ref{lemma:sym2Cond34} or \ref{lemma:sym3Cond34} hold, then Lemma \ref{lemma:sym4_Simple} also holds. This explains why no condition upon $F_j(\bm x)$ needs to be checked in Proposition~\ref{prop:switchings}. Thus the assumptions of Proposition~\ref{prop:switchings} result from solely combining the assumptions of Lemmas~\ref{lemma:sym1Cond34} and \ref{lemma:sym4_Simple}. Moreover, we take $\Omega_\tc{a}=\Omega_\tc{b}$ to be a condition of Proposition~\ref{prop:switchings} for simplicity, even though this condition is not needed to satisfy Lemma~\ref{lemma:sym1Cond34}. In general, the conditions 5 and 6 of the Proposition~\ref{prop:switchings} can be replaced by the following less restrictive combined condition: in the comoving frame, for some time $t_{\tc{r}}$, it is satified that either
\begin{align}
    &\chi_j(t+t_{\tc{r}})=\chi_j(t_{\tc{r}}-t),\ j=A,B,\label{eq:reverse_apx}\\
    &\text{or }\;\chi_\tc{a}(t+t_{\tc{r}})=\chi_\tc{b}(t-t_{\tc{r}}) \text{ and } \Omega_\tc{a}=\Omega_\tc{b}.\label{eq:swapandreverse_apx}
\end{align}


For future convenience, we point out that, in the scenarios where Proposition~\ref{prop:switchings} applies and both detectors start in the ground state, $\mathcal M$ simplifies to
\begin{align}
    \mathcal M&=-\lambda^2 \int \dd t \dd{t'} \dd{\bm x} \dd {\bm x'} e^{\ii \Omega (t + t')} \chi_\tc{a}(t)\chi_\tc{b}(t')\nonumber\\
    &\qquad\qquad\times F_\tc{a}(\bm x)F_\tc{b}(\bm x') G_{\tc{f}}(t,\bm x,t', \bm x').\label{eq:Mforprop2}
\end{align}
If both detectors start in the excited state, $\Omega$ in the expression above picks up a negative sign. 

\subsection{Sufficient symmetries for constructive interference between harvesting and communication in FRW spacetimes}
\label{apx:proveFRWprop}

Here, we show that Proposition~\ref{prop:switchingsFRW} for FRW spacetimes holds for the entanglement harvesting scenarios prescribed in the subsection \ref{sec:deSitter}. The proof arises from the fact that $\mathcal M$ can be rewritten to have an analogous form to the $\mathcal M$ term of setups where Proposition~\ref{prop:switchings} applies. Concretely, the only changes are the replacement given in Eq.~\eqref{eq:replacement} and the substitution of $t$ by $\eta(t)$ in $G_\tc{f}$. Moreover, using the assumption $\Lambda_j(t,\bm x)= \mathfrak{T}_j(t)\mathfrak{X}_j(\bm x)$ results in
\begin{align}
    \mathcal M &=-\lambda^2 \int \dd t \dd{t'} \dd{\bm x} \dd {\bm x'} e^{\ii \Omega (t + t')} \mathfrak{T}_\tc{a}(t)\mathfrak{T}_\tc{b}(t')\nonumber\\
    &\qquad\qquad\times \mathfrak{X}_\tc{a}(\bm x)\mathfrak{X}_\tc{b}(\bm x')\frac{G_{\tc{f}}\big(\eta(t),\bm x,\eta(t'), \bm x'\big)}{a(t)a(t')}.\label{eq:MforFRW?}
\end{align}
Here, $G_\tc{f}$ is the usual expression in inertial coordinates of the Feynmann propagator for the Minkowski vacuum, as given in the proof of Lemma~\ref{lemma:MinkowskiGF}. 
Notice that this expression for $\mathcal M$ assumes that both detectors start in the ground state. The case where both detectors are prepared in the excited state amounts to considering $\Omega < 0$. Therefore, this case is automatically included in Eq.~\eqref{eq:MforFRW?}, and in the current proof. 

Notice that if we define
\begin{equation}
    \widetilde{G}_\tc{f}(t,\bm x,t, \bm x')=\frac{G_{\tc{f}}\big(\eta(t),\bm x,\eta(t'), \bm x'\big)}{a(t)a(t')},
\end{equation}
then the expression for $\mathcal M$ becomes the same as for the scenarios of Proposition~\ref{prop:switchings} (see Eq.~\eqref{eq:Mforprop2}), but with $\widetilde{G}_\tc{f}$, $\mathfrak{T}_j(t)$, $\mathfrak{X}_j(\bm x)$ instead of $G_{\tc{f}}$, $\chi_j(t)$, $F_j(\bm x)$. Therefore, to get the analogous Proposition~\ref{prop:switchingsFRW}, we need $\widetilde{G}_\tc{f}$ and $\mathfrak{T}_j(t)$ to fulfill the corresponding conditions specified in the Symmetries~\ref{sym:tR} and \ref{sym:propertimeReversalSwap}, which were the ones used to derive Proposition~\ref{prop:switchings}. Notice that now, to define the Symmetry~\ref{sym:propertimeReversalSwap}, we need to use the comoving coordinates $(t,\bm x)$ as prescription coordinates for the detectors. This is a consequence of assuming that the split $\Lambda_j(t,\bm x)= \mathfrak{T}_j(t)\mathfrak{X}_j(\bm x)$ happens in the comoving coordinates, and that the detectors' proper times coincide with the comoving time $t$.

The conditions over $\mathfrak{T}_j(t)$ are directly stated in Proposition~\ref{prop:switchingsFRW} and are the same as for $\chi_j(t)$ in Proposition~\ref{prop:switchings}. Furthermore, since condition 2 of Symmetries~\ref{sym:tR} and \ref{sym:propertimeReversalSwap} was needed to prove Proposition~\ref{prop:switchings}, here we also need to verify whether $\widetilde{G}_\tc{f}$ satisfies to the same condition. Namely, we would like to show that
\begin{align}
    \widetilde{G}_\tc{f}(t,\bm x, t',\bm x')
    &=\widetilde{G}_\tc{f}(2 t_R-t,\bm x, 2 t_R-t',\bm x'),\nonumber\\
    &=\widetilde{G}_\tc{f}(2 t_R-t',\bm x, 2 t_R-t,\bm x').\label{eq:condsTildeGF}
\end{align}
To show that this relation is indeed fulfilled, first notice that the second equal sign can be readily checked by substituting the definition of $\widetilde{G}_\tc{f}$ and using that $G_\tc{f}$ fulfills the same equality, as shown in Lemma~\ref{lemma:MinkowskiGF}. Second, we need to show that $\widetilde{G}_\tc{f}$ is time symmetric around $t_\tc{R}$. For this proof, we start with the hypothesis of  Proposition~\ref{prop:switchingsFRW} that $a(t)=a(2t_\tc{r}-t)$. Then, the conformal time (as defined in Eq.~\eqref{eq:defConformalTime}) fulfills
\begin{equation}
    \eta(t)=2\eta(t_\tc{r})-\eta(2t_\tc{r}-t),
\end{equation}
which can be readily verified by taking the derivative with respect to $t$. Then,
\begin{align}
    &\widetilde{G}_\tc{f}(t,\bm x, t',\bm x')\nonumber\\
    &=\frac{G_{\tc{f}}\big(\eta(t),\bm x,\eta(t'), \bm x'\big)}{a(t)a(t')}\nonumber\\
    &=\frac{G_{\tc{f}}\big(2\eta(t_\tc{r})-\eta(2t_\tc{r}-t),\bm x,2\eta(t_\tc{r})-\eta(2t_\tc{r}-t'), \bm x'\big)}{a(2t_\tc{r}-t)a(2t_\tc{r}-t')}\nonumber\\
    &=\frac{G_{\tc{f}}\big(\eta(2t_\tc{r}-t),\bm x,\eta(2t_\tc{r}-t'), \bm x'\big)}{a(2t_\tc{r}-t)a(2t_\tc{r}-t')}\nonumber\\
    &=\widetilde{G}_\tc{f}(2t_\tc{r}-t,\bm x, 2t_\tc{r}-t',\bm x'),
\end{align}
where we used the first line of Eq.~\eqref{eq:symGF} for the Minkowski vacuum $G_\tc{f}$. This completes the proof that Eq.~\eqref{eq:condsTildeGF} holds under the hypothesis that \mbox{$a(t)=a(2t_\tc{r}-t)$}. Bringing everything together, an analogous version of Proposition~\ref{prop:switchings} holds for FRW spacetimes, resulting in Proposition~\ref{prop:switchingsFRW}.

\section{Alternative symmetry conditions for detectors prescribed in Fermi normal coordinates}
\label{apx:SymmFermiNormal}

\subsection{Conditions for the time reversibility and spatial reflection symmetry of Fermi normal coordinates}
\label{apx:FermiNormal}
Here we shall study how pairs of Fermi normal coordinates are related by a time reversal or a time reversal combined with a spatial reflection. Knowing the symmetries respected by the Fermi normal coordinates is instrumental in determining the symmetries of the entanglement harvesting setup. This is the case when the harvesting setup uses the usual covariant prescription for the Unruh-DeWitt detector~\cite{TalesBrunoEdu2020} which is based on the Fermi normal coordinates adapted to the detector trajectory $\mf z(\tau)$. Concretely, given an arbitrary coordinate system $(t,\bm x)$, we will consider the transformations:
\begin{equation}
    T_{\tc{r}}(t,\bm x) = (2t_{\tc{r}}-t,\bm x),\ T_\tc{rr}(t,\bm x) = (2t_{\tc{r}}-t,2\bm x_0-\bm x),\label{eq:TrTrr_apxdef}
\end{equation}
where time reversals are performed around the time $t_{\tc{r}}$, and the spatial reflection is performed about the point $\bm x_0$. The lemmas in this appendix provide the conditions on the metric and trajectories for the pair of Fermi normal coordinates to be related either by $T_{\tc{r}}$ or by $T_\tc{rr}$, which will be later used in Appendix~\ref{apx:FNCpropositions}.

First, we recall how to construct the Fermi normal coordinates $(\tau,\bm X)$ adapted to a timelike curve $\mf z(\tau)$, where $\tau$ is the curve's proper time. Choose a fixed proper time $\tau_0$, and pick an orthonormal basis
$\big\{\mf{e}_\mu(\tau_0)\big\}$ in the tangent space to $\mf z(\tau_0)$, such that $\mf{e}_0 = \dot{\mf z}(\tau_0)$. This basis extends to the whole curve, $\big\{\mf{e}_\mu(\tau)\big\}$, by imposing that the basis vectors are Fermi-Walker (FW)  transported along the curve. A vector field $v^\alpha$ on the curve is FW transported if
\begin{equation}
    \frac{\text{D} v^\alpha}{\dd \tau} + (a^\alpha u^\beta - u^\alpha a^\beta) v_\beta = 0,\label{eq:FW_transport}
\end{equation}
where $u^\alpha$ and $a^\alpha$ are the components of $\mf u = \dd \mf z /\dd \tau$ and \mbox{$\mf a = \text{D} \mf u /\dd \tau$}, the four-velocity and four-acceleration of the curve, respectively. Moreover, we use $\text{D}/\dd \tau$ to denote the directional covariant derivative along $\mf z(\tau)$, which for a vector $ \mf v = v^\alpha \partial_{\alpha}$ is
\begin{equation}
    \frac{\text{D} v^\alpha}{\dd \tau} = \frac{\dd v^\alpha}{\dd \tau} + \Gamma^\alpha_{\beta\gamma} v^\beta u^\gamma.
\end{equation}
Notably, $\mf u$ is always FW transported, therefore $\mf{e}_0(\tau) = \mf{u}(\tau)$. Moreover, $\big\{\mf{e}_\mu(\tau)\big\}$ remains an orthonormal basis for each point of the curve. The idea behind the definition of the FW transport is to generalize the parallel transport to account for the possible change in orientation introduced by the acceleration of the trajectory. That is, we would like to transport the frame $\{\mf e_{\mu}(\tau)\}$ in a way that it does not rotate with respect to the curve. To see this concretely, consider coordinates where $u^{\mu} = (u^{0}, 0, 0, 0)$, and define the two-form  $\omega_{\alpha \beta} = 2a_{[\alpha}u_{\beta]}$. In this case, the only non-vanishing components of $\omega_{\alpha \beta}$ are $\omega_{i0} = -\omega_{0i}$, $i=1, 2, 3$, since $a^{\mu}u_{\mu} = 0$. Thus, the second term in Eq.~\eqref{eq:FW_transport} is essentially describing the projection of $v^{\mu}$ along the $3$-vector $n^{i} = \omega^{i0}$, which describes the rotation of the orthonormal ``spatial'' frame $\mf e_{i}$ along the trajectory.

Now, we define the Fermi normal coordinates as follows. Let $\mathcal N_{\mf p}$ be the normal neighborhood of $\mf p$, which is the set of points that are connected to $\mf p$ by a unique geodesic. For each $\tau$, denote with $\Sigma_\tau$ the space-like hypersurface orthogonal to $\mf{u}(\tau)$. This hypersurface consists of all the points in $\mathcal N_{\mf z(\tau)}$ that can be reached by geodesics that start from $\mf z(\tau)$ and start with tangent vectors orthogonal to $\mf{u}(\tau)$. These hypersurfaces $\Sigma_\tau$ constitute rest spaces around $\mf z(\tau)$, and locally define a foliation of spacetime. Then, the coordinates $(\tau,\bm X)$ are assigned to the point $\exp_{\mf z(\tau)}\big(X^a \mf{e}_a(\tau)\big)\in \Sigma_\tau$. Here, we use the convention that latin indices run over spatial components, $a=1,\ldots,n$. The exponential map is defined as
\begin{equation}
    \exp_{\mf p}\big( \mf v\big) = \gamma_{\mf v}(1),
\end{equation}
where $\gamma_{\mf v}(s)$ is a geodesic that fulfills $\gamma_{\mf v}(0)= \mf p$, \mbox{$\dot\gamma_{\mf v}(0)= \mf v$}. Notice that in this prescription $\mf z(\tau)$ has coordinates $(\tau, \bm 0)$, and its proper distance to a point with coordinates $(\tau,\bm X)$ is $|\bm X|=\sqrt{\sum_a (X^a)^2}$.

\begin{lemma} \label{lemma:FNreversible}
    Choose arbitrary $(t, \bm x)$ coordinates, in which the metric components are $g_{\mu\nu}$. Consider a pair of time-like curves $\mf z_j(\tau_j)$, $j=A,B$, parametrized by their proper times $\tau_j$. Denote the coordinates of the curves' points as $\big(t,\bm z_j(t)\big)$. Assume that for some $t_{\tc{r}}$,
    \begin{enumerate}
        \item $g_{\mu\nu}(t_{\tc{r}}+t,\bm x) = s(\mu, \nu) g_{\mu\nu}(t_{\tc{r}}-t,\bm x)$,
        \item $\bm z_\tc{a}(t_{\tc{r}}+t)=\bm z_\tc{b}(t_{\tc{r}}-t)$,
    \end{enumerate}
    Here, $s(\mu_1, \mu_{2}, \ldots, \mu_{n})= 1$ if the number of zeros among the set of arguments $\mu_{1}, \ldots, \mu_{n}$ is even and $s(\mu_1, \mu_{2}, \ldots, \mu_{n})= -1$ otherwise.
    Then, there exist Fermi normal coordinates $(\tau_j,\widetilde{\bm x}_j)$ adapted to the trajectories $\mf z_j(\tau_j)$ that are related by
    \begin{align}
        \tau_\tc{a}(t_{\tc{r}}+t, \bm x)&  = \upsilon_{\tc{r}} -\tau_\tc{b}(t_{\tc{r}}-t, \bm x), \nonumber\\
        \widetilde{\bm x}_\tc{a}(t_{\tc{r}}+t, \bm x)& = \widetilde{\bm x}_\tc{b}(t_{\tc{r}}-t, \bm x),\label{eq:RelateFNC}
    \end{align}
    with $\upsilon_\tc{r} = \tau_{{\tc{r}},\tc{a}}+\tau_{\tc{r},\tc{b}}$, 
    where $\tau_{\tc{r},j}$ are the proper times that fulfill $t_{\tc{r}} = z^0_j(\tau_{\tc{r},j})$.
    
\end{lemma}

\begin{proof}
    Let us start defining a convenient notation: for any function $f(t)$, let
    \begin{equation}
        [f(t)]_{{\tc{r}}} = f(2t_\tc{r} -t).
    \end{equation}
    Then, the assumptions of the lemma can be simply rewritten as
    \begin{equation}
          [g_{\mu\nu}]_{\tc{R}} = s(\mu,\nu) g_{\mu\nu},\  [{\bm z}_\tc{a}]_{\tc{r}}=[\bm z_\tc{b}]_{\tc{r}}.
    \end{equation}
    Next, let us show that these assumptions imply that
    \begin{enumerate}[label=\roman*)]
        \item \label{it:inv} $[g^{\mu\nu}]_{\tc{r}} = s(\mu,\nu)g^{\mu\nu}$, 
        \item \label{it:derg} $ [\partial_\lambda g_{\mu\nu}]_{\tc{r}} = s(\mu, \nu, \lambda)\partial_\lambda g_{\mu\nu}$,
        \item \label{it:chris} $[\Gamma_{\mu\nu}^\lambda]_{\tc{r}} = s(\mu, \nu, \lambda)\Gamma_{\mu\nu}^\lambda$,
        \item \label{it:speed} $[ u^\mu_\tc{a}]_{\tc{r}} = -s(\mu) u^\mu_\tc{b}$,
        \item \label{it:acc} $[a^\mu_\tc{a}]_{\tc{r}} = s(\mu) a^\mu_\tc{b}$.
    \end{enumerate}
    Showing that $s(\mu,\nu)g^{\mu\nu}$ is the inverse of $[g_{\mu\nu}]_{\tc{r}}$ is enough for \ref{it:inv},
    \begin{equation}
        s(\mu,\nu)g^{\mu\nu}[ g_{\nu\lambda}]_{\tc{r}} = s(\mu,\lambda) g^{\mu\nu}g_{\nu\lambda} = \delta^\lambda_\mu.
    \end{equation}
    For \ref{it:derg}, it is enough to use \ref{it:inv} and to notice that for a general function $h(t,\bm x)$,
    \begin{equation}
        \partial_0[h]_{\tc{r}}= - [\partial_0 h]_{\tc{r}},\ \partial_a[h]_{\tc{r}}= [\partial_a h]_{\tc{r}}.\label{eq:timeDerRel}
    \end{equation}
    Relation \ref{it:chris} follows from combining \ref{it:inv} and \ref{it:derg} with the formula for the Christoffel symbols, namely
    \begin{equation}
        \Gamma_{\mu\nu}^\lambda=\frac{1}{2} g^{\lambda \sigma} \left(\partial_{\nu}g_{\sigma \mu} + \partial_{\mu}g_{\sigma \nu} - \partial_{\sigma}g_{\mu \nu}\right).
    \end{equation}
    Proving \ref{it:speed} starts from taking the derivative on the trajectory condition, to obtain $[{\dot {\bm z}}_\tc{a}]_{\tc{r}}=-\dot{\bm z}_\tc{b}$. Then, using
    \begin{equation}
        u_j^0 = \frac{\dd t}{\dd \tau_j},\ u_j^a =  \frac{\dd t}{\dd \tau_j} \dot {\bm z}_j^a,\label{eq:speedsRel}
    \end{equation}
    together with $g_{\mu\nu} u_j^\mu u_j^\nu  =-1$, implies
    \begin{equation}
        \frac{\dd t}{\dd \tau_j} = \sqrt{-g_{00}-2g_{a0}\dot{\bm z}_j^a-g_{ab}\dot{\bm z}_j^a\dot{\bm z}_j^b}.
    \end{equation}
    Therefore, using $[{\dot {\bm z}}_\tc{a}]_{\tc{r}}=-\dot{\bm z}_\tc{b}$ and \ref{it:derg}, we get
    \begin{equation}
        \left[\frac{\dd t}{\dd \tau_\tc{a}}\right]_{\tc{r}} = \frac{\dd t}{\dd \tau_\tc{b}}.\label{eq:gammaFactorRel}
    \end{equation}
    Then, substituting back on Eq.~\eqref{eq:speedsRel} allows showing \ref{it:speed}.
    Finally, to prove \ref{it:acc}, combine 
    \begin{equation}
        a_j^\alpha = \frac{\dd t}{\dd \tau_j}\frac{\dd u_j^\alpha}{\dd t} + \Gamma^\alpha_{\beta\gamma} u_j^\beta u_j^\gamma
    \end{equation}
    with Eqs.~\eqref{eq:timeDerRel},~\eqref{eq:gammaFactorRel}, and relations \ref{it:chris}, \ref{it:speed}.
    
    Next, we prove that the FW transported orthonormal bases along the curves can be related by a time reversal transformation.  Consider an orthonormal basis $\{\mf{e}_\mu(t)\}$ that is FW transported along $\mf z_\tc{a}$ and fulfills $\mf{e}_0=\mf{u}_\tc{a}$. Define $\mf{f}_\mu$ such that
    \begin{equation}
        (\mf{f}_\mu)^\alpha = s(\mu,\alpha)[(\mf{e}_\mu)^\alpha]_{\tc{r}}.\label{eq:ONbasisRel}
    \end{equation}
    We now will show that this basis is FW transported along $\mf z_\tc{b}$.
    Let us start reversing the FW transport equation that the $\{\mf{e}_\mu\}$ fulfill,
    \begin{align}
        0&=\left[\frac{\text{D} (\mf{e}_\mu)^\alpha}{\dd \tau_\tc{a}} + (a_\tc{a}^\alpha u_\tc{a}^\beta - u_\tc{a}^\alpha a_\tc{a}^\beta) g_{\beta\gamma}(\mf{e}_\mu)^\gamma\right]_{\tc{r}} \nonumber\\
        &=-s(\mu,\alpha) \bigg(\frac{\text{D} (\mf{f}_\mu)^\alpha}{\dd \tau_\tc{b}} + (a_\tc{b}^\alpha u_\tc{b}^\beta - u_\tc{b}^\alpha a_\tc{b}^\beta) g_{\beta\gamma}(\mf{f}_\mu)^\gamma\bigg),
    \end{align}
    where we used the relations \ref{it:derg}, \ref{it:speed}, \ref{it:acc}, and also
    \begin{align}
        \left[\frac{\text{D} (\mf{e}_\mu)^\alpha}{\dd \tau_\tc{a}}\right]_{\tc{r}} &= \left[\frac{\dd t}{\dd \tau_\tc{a}}\frac{\dd (\mf{e}_\mu)^\alpha}{\dd t} + \Gamma^\alpha_{\beta\gamma} (\mf{e}_\mu)^\beta u_\tc{a}^\gamma \right]_{\tc{r}}\nonumber\\
        &=-s(\mu,\alpha)\frac{\text{D} (\mf{f}_\mu)^\alpha}{\dd \tau_\tc{b}},
    \end{align}
    which used Eqs.~\eqref{eq:timeDerRel},~\eqref{eq:gammaFactorRel},~\eqref{eq:ONbasisRel}, and relations \ref{it:chris}, \ref{it:speed}. Therefore, the $\{\mf{f}_\mu\}$, as defined in Eq.~\eqref{eq:ONbasisRel} are FW transported along $\mf z_\tc{b}$. Moreover, the $\{\mf{f}_\mu\}$ are orthonormal because
    \begin{equation}
        (\mf{f}_\mu)^\alpha (\mf{f}_\nu)^\beta g_{\alpha\beta} = \left[(\mf{e}_\mu)^\alpha (\mf{e}_\nu)^\beta g_{\alpha\beta}\right]_{\tc{r}} = \eta_{\mu\nu},
    \end{equation}
    where $\eta_{\mu\nu}$ is the  Minkowski metric. Lastly, $\mf{f}_0 = \mf{u}_\tc{b}$, because 
    \begin{equation}
        (\mf{f}_0)^\alpha = -s(\alpha)[(\mf{e}_0)^\alpha]_{\tc{r}} = -s(\alpha) \left[ u_\tc{a}^\alpha \right]_{\tc{r}}= u_\tc{b}^\alpha.
    \end{equation}

    All these results guarantee that the orthonormal basis $\{\mf{f}_\mu(t)\}$ can be used to define the Fermi normal coordinates along $\mf z_\tc{b}$. From now on, let us use $\{\mf{e}_\mu\}$ and $\{\mf{f}_\mu\}$ to respectively define the Fermi normal coordinates $(\tau_\tc{a},\widetilde{\bm x}_\tc{a})$ around $\mf z_\tc{a}(\tau_\tc{a})$ and $(\tau_\tc{b},\widetilde{\bm x}_\tc{b})$ around $\mf z_\tc{b}(\tau_\tc{b})$. 

    Next, we check the relationships in Eq.~\eqref{eq:RelateFNC} between $(\tau_\tc{a},\widetilde{\bm x}_\tc{a})$ and $(\tau_\tc{b},\widetilde{\bm x}_\tc{b})$. 
    For any spacetime point $\mf p_\tc{a}$, we define $\mf p_\tc{b}$ as follows. 
    We denote the Fermi normal coordinates $j$ of the point $\mf p_j$ as
    \begin{equation}
         \tau_{\mf p_j}=\tau_j(\mf p_j),\ \widetilde{\bm x}_{\mf p_j}=\widetilde{\bm x}_j(\mf p_j).\label{eq:shortennotationpj}
    \end{equation}
    Then, $\mf p_\tc{b}$ is chosen so that the coordinates of $\mf p_\tc{a}$ and $\mf p_\tc{b}$ are related by
    \begin{equation}
        \tau_{\mf p_\tc{a}} = \upsilon_{\tc{r}} -\tau_{\mf p_\tc{b}},\ \widetilde{\bm x}_{\mf p_\tc{a}}=\widetilde{\bm x}_{\mf p_\tc{b}}, \label{eq:properTimesRel}
    \end{equation}
    with $\upsilon_{\tc{r}} = \tau_{\tc{r},\tc{a}}+\tau_{\tc{r},\tc{b}}$, where $\tau_{\tc{r},j}$ are the proper times that fulfill $t_{\tc{r}} = z^0_j(\tau_{\tc{r},j})$. 
    By expressing the points $\mf p_j$ in the coordinate system $(t,\bm x)$ on Eqs.~\eqref{eq:shortennotationpj} and~\eqref{eq:properTimesRel}, one can see that the Eq.~\eqref{eq:RelateFNC} that we want to prove is equivalent to showing $T_\tc{r}(\mf p_\tc{a})=\mf p_\tc{b}$, with $T_\tc{r}$ as defined in \eqref{eq:TrTrr_apxdef}. 
    Therefore, to complete the proof, it will be enough to show that $T_\tc{r}(\mf p_\tc{a})=\mf p_\tc{b}$ is true, as is done next. 
    
    Consider a geodesic $\gamma_\tc{a}(\lambda)$ such that
    \begin{equation}
         \gamma_\tc{a}(0)=\mf z_\tc{a}(\tau_{\mf p_\tc{a}}),\ 
        \dot\gamma_\tc{a}(0) = \widetilde{x}_{\mf p_\tc{a}}^a \mf{e}_a\big(z^0_\tc{a}(\tau_{\mf p_\tc{a}})\big).
    \end{equation}
    Notice that, by the definition of the Fermi normal coordinates, we have $\gamma_\tc{a}(1)=\mf p_\tc{a}$.
    Next, define $\gamma_\tc{b}(\lambda)$ as
    \begin{equation}
        \gamma_\tc{b}(\lambda) = T_\tc{r}(\gamma_\tc{a}(\lambda)).\label{eq:defgammab}
    \end{equation}
    The next steps are to show that such $\gamma_\tc{b}(\lambda)$ is a geodesic that fulfills 
    \begin{equation}
        \gamma_\tc{b}(0)=\mf z_\tc{b}(\tau_{\mf p_\tc{b}}),\  \dot\gamma_\tc{b}(0)=\widetilde{x}_{\mf p_\tc{b}}^a \mf{f}_a\big(z^0_\tc{b}(\tau_{\mf p_\tc{b}})\big).\label{eq:gammabObjective}
    \end{equation}
    We start checking the first identity in the equation above by computing
    \begin{align}
        \frac{\dd}{\dd \tau}\Big(\big(\mf z_\tc{b}^0\big)^{-1}\big(2 t_{\tc{r}} - z^0_\tc{a}(\tau)\big)\Big)
        &=-\frac{u_\tc{a}^0\big(z^0_\tc{a}(\tau)\big)}{u_\tc{b}^0\big(2 t_{\tc{r}} - z^0_\tc{a}(\tau)\big)}&=-1,
    \end{align}
    where we used the statement \ref{it:speed} that we proved at the start of the proof. Therefore, for some  constant $\tau_0$,
    \begin{equation}
        z_\tc{b}^0(\tau_0-\tau) = 2 t_{\tc{r}} - z^0_\tc{a}(\tau). \label{eq:relzbza}
    \end{equation}
    To find $\tau_0$, choose $\tau=\tau_{\tc{r},\tc{a}}$, so that
    \begin{equation}
        z_\tc{b}^0(\tau_0-\tau_{\tc{r},\tc{a}}) = 2 t_{\tc{r}} - z^0_\tc{a}(\tau_{\tc{r},\tc{a}})=t_{\tc{r}},
    \end{equation}
    which means $\tau_0-\tau_{\tc{r},\tc{a}}=\tau_{\tc{r},\tc{b}}$, and thus \mbox{$\tau_0=\upsilon_\tc{r}$}. Substituting back to Eq.~\eqref{eq:relzbza}, and using $\tau_{\mf p_\tc{a}} = \upsilon_\tc{r} -\tau_{\mf p_\tc{b}}$,
    \begin{equation}
        z_\tc{b}^0(\tau_{\mf p_\tc{b}}) = 2 t_{\tc{r}} - z^0_\tc{a}(\tau_{\mf p_\tc{a}}),\label{eq:reverseza0zb0}
    \end{equation}
    and therefore, $\gamma_\tc{b}^0(0)= z_\tc{b}^0(\tau_{\mf p_\tc{b}})$. Moreover, we have \mbox{$\gamma_\tc{b}^a(0) = z_\tc{b}^a(\tau_{\mf p_\tc{b}})$} because $\gamma_\tc{b}^a(0)= z_\tc{a}^a(\tau_{\mf p_\tc{a}})$ and
    \begin{align}
        \bm z_\tc{a}\big(z^0_\tc{a}(\tau_{\mf p_\tc{a}})\big)=\bm z_\tc{b}\big(2t_{\tc{r}}-z^0_\tc{a}(\tau_{\mf p_\tc{a}})\big)=\bm z_\tc{b}\big(z^0_\tc{b}(\tau_{\mf p_\tc{b}})\big),
    \end{align}
    where we used the assumption $\bm z_\tc{a}(t)=\bm z_\tc{b}(2t_{\tc{r}}-t)$. This completes showing the claim from Eq.~\eqref{eq:gammabObjective} that \mbox{$\gamma_\tc{b}(0)=\mf z_\tc{b}(\tau_{\mf p_\tc{b}})$}.

    To get the identity for $\dot\gamma_\tc{b}(0)$ in Eq.~\eqref{eq:gammabObjective}, we start from the definition of $\gamma_\tc{b}$ given in Eq.~\eqref{eq:defgammab}. Expressing this definition in $(t,\bm x)$ coordinates gives
    \begin{equation}
        \gamma_\tc{b}^0(\lambda)=2t_\tc{r} - \gamma_\tc{a}^0(\lambda),\ \gamma_\tc{b}^a(\lambda)=\gamma_\tc{a}^a(\lambda).
    \end{equation}
    Then, taking derivatives w.r.t. $\lambda$, we get
    \begin{equation}
        \dot\gamma_\tc{b}^\alpha(\lambda) = s(\alpha) \dot\gamma_\tc{a}^\alpha(\lambda). \label{eq:derivative_gammab}
    \end{equation}
    Combining this result with Eqs.~\eqref{eq:ONbasisRel}, \eqref{eq:properTimesRel} and~\eqref{eq:reverseza0zb0} leads to
    \begin{align}
        \dot\gamma_\tc{b}^\alpha(0)&= s(\alpha)\dot\gamma_\tc{a}^\alpha(0)\nonumber\\
        &=s(\alpha)\widetilde{x}_{\mf p_\tc{a}}^a (\mf{e}_a)^\alpha\big(z^0_\tc{a}(\tau_{\mf p_\tc{a}})\big)\nonumber\\
        &=\widetilde{x}_{\mf p_\tc{b}}^a (\mf{f}_a)^\alpha\big(z^0_\tc{b}(\tau_{\mf p_\tc{b}})\big),
    \end{align}
    which, as we wanted to show, is the expression for $\dot\gamma_\tc{b}(0)$ in Eq.~\eqref{eq:gammabObjective}.
    
    Finally, let us show that $\gamma_\tc{b}$ is a geodesic. 
    Taking the derivative w.r.t. $\lambda$ in Eq.~\eqref{eq:derivative_gammab},
    \begin{equation}
        \dot\gamma_\tc{b}^\alpha(\lambda) = s(\alpha) \dot\gamma_\tc{a}^\alpha(\lambda),\ \ddot\gamma_\tc{b}^\alpha(\lambda) =s(\alpha) \ddot\gamma_\tc{a}^\alpha(\lambda).
    \end{equation}
    Then, substituting these expressions into the geodesic equation which $\gamma_\tc{a}$ fulfills, 
    \begin{align}
        \ddot \gamma_\tc{a}^\alpha &= -\Gamma^\alpha_{\beta\gamma}\big(\gamma^0_\tc{a},\gamma^a_\tc{a}\big)\,\dot \gamma_\tc{a}^\beta\dot \gamma_\tc{a}^\gamma,\nonumber\\
        s(\alpha)\ddot \gamma_\tc{b}^\alpha &= -s(\alpha)\Gamma^\alpha_{\beta\gamma}\big(2t_{\tc{r}}-\gamma^0_\tc{a},\gamma^a_\tc{a}\big)\,\dot \gamma_\tc{b}^\beta\dot \gamma_\tc{b}^\gamma,\nonumber\\
        \ddot \gamma_\tc{b}^\alpha &= -\Gamma^\alpha_{\beta\gamma}\big(\gamma^0_\tc{b},\gamma^a_\tc{b}\big)\,\dot \gamma_\tc{b}^\beta\dot \gamma_\tc{b}^\gamma,
    \end{align}
    where we also used the statement \ref{it:chris} and Eq.~\eqref{eq:defgammab}. The last line is, as we wanted, the geodesic equation for $\gamma_\tc{b}$.
    
    Therefore, $\gamma_\tc{b}$ matches the requirements of the exponential map definition and thus $\mf p_\tc{b}=\gamma_\tc{b}(1)$. Now, remember that $T_\tc{r}(\gamma_\tc{a}(\lambda))=\gamma_\tc{b}(\lambda)$. Then, choosing $\lambda=1$, $T_\tc{r}(\mf p_\tc{a})=\mf p_\tc{b}$. This finally completes the proof that 
    \begin{align}
        \tau_\tc{a}(t_{\tc{r}}+t, \bm x) = \upsilon_\tc{r} -\tau_\tc{b}(t_{\tc{r}}-t, \bm x), \nonumber\\
        \widetilde{\bm x}_\tc{a}(t_{\tc{r}}+t, \bm x) = \widetilde{\bm x}_\tc{b}(t_{\tc{r}}-t, \bm x).
    \end{align}
\end{proof}

\begin{lemma}\label{lemma:SingleFNCreversible}
    Choose arbitrary coordinates $(t, \bm x)$, in which the metric components are $g_{\mu\nu}$. Consider a time-like curve $\mf z(\tau)$, parametrized by its proper time $\tau$. Denote the coordinates of the curve's points as $\big(t,\bm z(t)\big)$. Assume that for some $t_{\tc{r}}$,
    \begin{enumerate}
        \item $g_{\mu\nu}(t_{\tc{r}}+t,\bm x) = s(\mu,\nu) g_{\mu\nu}(t_{\tc{r}}-t,\bm x)$,
        \item $\bm z(t_{\tc{r}}+t)=\bm z(t_{\tc{r}}-t)$,
    \end{enumerate}
    Here, $s(\mu_1, \mu_{2}, \ldots, \mu_{n})= 1$ if the number of zeros among the set of arguments $\mu_{1}, \ldots, \mu_{n}$ is even and $s(\mu_1, \mu_{2}, \ldots, \mu_{n})= -1$ otherwise.
    Then, the Fermi normal coordinates $(\tau,\widetilde{\bm x})$ adapted to the trajectory $\mf z(\tau)$ fulfill
    \begin{align}
        \tau(t_{\tc{r}}+t, \bm x)&  = 2\tau_{\tc{r}} -\tau(t_{\tc{r}}-t, \bm x), \nonumber\\
        \widetilde{\bm x}(t_{\tc{r}}+t, \bm x)& = \widetilde{\bm x}(t_{\tc{r}}-t, \bm x),\label{eq:RelateSingleFNC}
    \end{align}
    where $\tau_{\tc{r}}$ is the proper time that fulfills $t_{\tc{r}} = z^0(\tau_{\tc{r}})$.
\end{lemma}
\begin{proof}
   The assumptions of this lemma fulfill the hypotheses of Lemma~\ref{lemma:FNreversible} for the case $ \mf z(\tau)= \mf z_\tc{a}(\tau)=\mf z_\tc{b}(\tau)$. Then, Lemma~\ref{lemma:FNreversible} provides two coordinate systems $(\tau_\tc{a},\widetilde{\bm x}_\tc{a})$ and $(\tau_\tc{b},\widetilde{\bm x}_\tc{b})$ adapted to $\mf z(\tau)$, which are related by the Eq.~\eqref{eq:RelateFNC}. A priori, these two coordinate systems could be different, because there is freedom in the choice of orthonormal basis that defines any Fermi normal coordinates adapted to $\mf z(\tau)$. However, Eq.~\eqref{eq:RelateFNC} implies that
   \begin{align}
       \widetilde{\bm x}_\tc{a}(t_{\tc{r}}, \bm x)& = \widetilde{\bm x}_\tc{b}(t_{\tc{r}}, \bm x),
   \end{align}
   which means that there is a spatial surface where both Fermi normal coordinates $(\tau_\tc{a},\widetilde{\bm x}_\tc{a})$ and $(\tau_\tc{b},\widetilde{\bm x}_\tc{b})$ adapted to $\mf z(\tau)$ are the same. This can only happen if $(\tau_\tc{a},\widetilde{\bm x}_\tc{a})$ and $(\tau_\tc{b},\widetilde{\bm x}_\tc{b})$ are the same coordinate systems, thus proving that there are Fermi normal coordinates adapted to $\mf z(\tau)$ that fulfill Eq.~\eqref{eq:RelateSingleFNC}, finishing the proof.

\end{proof}

\begin{lemma} \label{lemma:FNreversibleReflect}
    Choose arbitrary coordinates $(t, \bm x)$, in which the metric components are $g_{\mu\nu}$. Consider a pair of time-like curves $\mf z_j(\tau_j)$, $j=A,B$, parametrized by their proper times $\tau_j$. Denote the coordinates of the curves' points as $\big(t,\bm z_j(t)\big)$. Assume that for some $t_{\tc{r}}$,
    \begin{enumerate}
        \item $g_{\mu\nu}(t_{\tc{r}}+t,\bm x_0+\bm x) = g_{\mu\nu}(t_{\tc{r}}-t,\bm x_0-\bm x)$,
        \item $\bm z_\tc{a}(t_{\tc{r}}+t)=2\bm x_0-\bm z_\tc{b}(t_{\tc{r}}-t)$.
    \end{enumerate}
    Then, there exist Fermi normal coordinates $(\tau_j,\widetilde{\bm x}_j)$ adapted to the trajectories $\bm z_j(t)$ that are related by
    \begin{align}
        \tau_\tc{a}(t_{\tc{r}}+t, \bm x_0+\bm x) = \upsilon_\tc{r} -\tau_\tc{b}(t_{\tc{r}}-t, \bm x_0-\bm x), \nonumber\\
        \widetilde{\bm x}_\tc{a}(t_{\tc{r}}+t, \bm x_0+\bm x) = \widetilde{\bm x}_\tc{b}(t_{\tc{r}}-t, \bm x_0-\bm x),
    \end{align}
     with $\upsilon_\tc{r} = \tau_{\tc{r},\tc{a}}+\tau_{\tc{r},\tc{b}}$, where $\tau_{\tc{r},j}$ are the proper times that fulfill $t_{\tc{r}} = z^0_j(\tau_{\tc{r},j})$.
\end{lemma}
\begin{proof}
    This lemma is proven analogously to Lemma~\ref{lemma:FNreversible}, so we will only point out the differences. Here, it is convenient to instead use the following notation for any functions $f(t)$, $h(t,\bm x)$,
    \begin{align}
        &[f(t)]_{\tc{r}} = f(2t_{\tc{r}}-t),\nonumber\\
        &[h(t,\bm x)]_{\tc{rr}} = h(2t_{\tc{r}}-t,2\bm x_0-\bm x).
    \end{align}
    Then, following similar steps to the previous proof,
    \begin{enumerate}[label=\roman*)]
        \item $[g_{\mu\nu}]_{\tc{r}} = g_{\mu\nu}$,
        \item $[\partial_\lambda g_{\mu\nu}]_{\tc{rr}} = -\partial_\lambda g_{\mu\nu}$,
        \item $[\Gamma_{\mu\nu}^\lambda]_{\tc{rr}} = -\Gamma_{\mu\nu}^\lambda$,
        \item $[u^\mu_\tc{a}]_{\tc{r}} = u^\mu_\tc{b}$,
        \item $[a^\mu_\tc{a}]_{\tc{r}} = - a^\mu_\tc{b}$.
    \end{enumerate}
    Notice that these relations are useful because the trajectories transform into one another when performing a time reversal and a spatial reflection. These relations can be used to find a pair of orthonormal basis $\{\mf{e}_\mu(t)\}$ and $\{\mf{f}_\mu(t)\}$, which respectively are FW transported along $\mf z_\tc{a}(\tau_\tc{a})$ and $\mf z_\tc{b}(\tau_\tc{b})$, and that fulfill
    \begin{equation}
        [(\mf{e}_\mu)^\alpha]_{\tc{r}} =  (\mf{f}_\mu)^\alpha. \label{eq:ONbasisRelwithRefl}
    \end{equation}
    Then, we again pick $\{\mf{e}_\mu(t)\}$ and $\{\mf{f}_\mu(t)\}$ to respectively define the Fermi normal coordinates adapted to $\mf z_\tc{a}(\tau_\tc{a})$ and $\mf z_\tc{b}(\tau_\tc{b})$. The final part of the proof is also analogous to Lemma~\ref{lemma:FNreversible}, but changing Eq.~\eqref{eq:defgammab} to relate $\gamma_\tc{a}$ and $\gamma_\tc{b}$ by $T_\tc{rr}$ instead of $T_\tc{r}$.
\end{proof}

\subsection{Simplified conditions to satisfy Corollary \ref{cor:symmetriesT} when detectors are prescribed in the Fermi normal coordinates}
\label{apx:FNCpropositions}

Here we combine the results obtained in the previous Subappendix with the results of the Appendices~\ref{apx:ExaSym} and \ref{apx:symcoordinates} to get simplified propositions that guarantee no destructive interference between harvesting and communication. This section assumes that we follow the covariant UDW detector prescription of \cite{TalesBrunoEdu2020}, where the detector $j$ is prescribed in the Fermi normal coordinates $(\tau_j,\widetilde{\bm x}_j)$ adapted to its trajectory $\mf z_j(\tau_j)$. These coordinates are valid locally around the trajectory and factorize the smearing function, \mbox{$\Lambda_j(\tau_j,\widetilde{\bm x}_j)=\chi_j(\tau_j)F_j(\widetilde{\bm x}_j)$}. Moreover, $\tau_j$ are the proper times of the detectors. 

\begin{prop} \label{prop:finalForm?}
Let $g_{\mu \nu}$ be the metric components in an arbitrary coordinate system $(t,\bm x)$, $\bm z_j(t)$ the spatial components of the detectors' trajectories and $(\tau_j,\widetilde{\bm x}_j)$ the Fermi normal coordinates adapted to $\bm z_j(t)$. Assume that
    \begin{enumerate}
        \item $g_{\mu\nu}(t_{\tc{r}}-t,\bm x) = s(\mu,\nu)g_{\mu\nu}(t_{\tc{r}}+t,\bm x)$,
        \item $G_\tc{f}(t_{\tc{r}}-t,\bm x, t_{\tc{r}}-t',\bm x')=G_\tc{f}(t_{\tc{r}}+t,\bm x, t_{\tc{r}}+t',\bm x')$,
        \item $\bm z_j(t_{\tc{r}}-t)=\bm z_j(t_{\tc{r}}+t)$,
        \item $\Lambda_j(\tau_j,\widetilde{\bm x}_j)=\chi_j(\tau_j)F_j(\widetilde{\bm x}_j)$,
        \item $\chi_j(\tau_{\tc{r},j}-\tau_j)=\chi_j(\tau_{\tc{r},j}+\tau_j)$,
        \item $\beta_j = \Omega_j\tau_{\tc{r},j}$,
    \end{enumerate}
    where $\tau_{\tc{r},j} = \tau_j(t_{\tc{r}}, \bm z_j(t_{\tc{r}}))$, and $s(\mu,\nu)$ is the sign function defined in Lemma~\ref{lemma:SingleFNCreversible}. Moreover, recall that the detectors' initial states are $\ket{\psi_j}=\cos\alpha_j\ket{g_j}+\sin\alpha_j e^{\ii \beta_j}\ket{e_j}$.
    
    These assumptions are sufficient to fulfill the conditions given in Symmetry \ref{sym:tR}, thus fulfilling Corollary~\ref{cor:symmetriesT}.
\end{prop}
This proposition follows from fulfilling the hypotheses of Lemma~\ref{lemma:SingleFNCreversible} in order to get the time reversibility condition over the Fermi normal coordinates in Eq.~\eqref{eq:hypLemmaLmu1} of Lemma~\ref{lemma:LambdaandMu} and, in turn, fulfill conditions 3 and 4 of Symmetry~\ref{sym:tR}. Moreover, notice that condition 1 of Symmetry~\ref{sym:tR}, i.e., $g(t_{\tc{r}}-t,\bm x)=g(t_{\tc{r}}+t,\bm x)$, is fulfilled due to the assumption of the proposition on the components $g_{\mu\nu}$.


\begin{prop} \label{prop:finalForm2}
    Let $g_{\mu \nu}$ be the metric components in an arbitrary coordinate system $(t,\bm x)$, $\bm z_j(t)$ the spatial components of the detectors' trajectories and $(\tau_j,\widetilde{\bm x}_j)$ the Fermi normal coordinates adapted to $\bm z_j(t)$. Assume that
    \begin{enumerate}
        \item $g_{\mu\nu}(t_{\tc{r}}-t,\bm x) = s(\mu,\nu)g_{\mu\nu}(t_{\tc{r}}+t,\bm x)$,
        \item $G_\tc{f}(t_{\tc{r}}-t,\bm x, t_{\tc{r}}-t',\bm x')=G_\tc{f}(t_{\tc{r}}+t',\bm x', t_{\tc{r}}+t,\bm x)$,
        \item $\bm z_\tc{a}(t_{\tc{r}}-t)=\bm z_\tc{b}(t_{\tc{r}}+t)$,
        \item $\widetilde{\bm x}_\tc{a}(t_{\tc{r}}, \bm x) = \widetilde{\bm x}_\tc{b}(t_{\tc{r}}, \bm x)$,
        \item $\Lambda_j(\tau_j,\widetilde{\bm x}_j)=\chi_j(\tau_j)F_j(\widetilde{\bm x}_j)$,
        \item $\chi_\tc{a}(\tau)=\chi_\tc{b}(\upsilon_\tc{r}-\tau)$,
        \item $F_\tc{a}(\widetilde{\bm x})=F_\tc{b}(\widetilde{\bm x})$,
        \item $\Omega_\tc{a} = \Omega_\tc{b}=\Omega$,
        \item $\beta_\tc{a}+\beta_\tc{b}=\Omega \upsilon_\tc{r}$, 
        \item $|\cos \alpha_\tc{a}|=|\cos \alpha_\tc{b}|$, $|\sin \alpha_\tc{a}|=|\sin \alpha_\tc{b}|$,
    \end{enumerate}
    where $\upsilon_\tc{r} = \tau_\tc{a}(t_{\tc{r}}, \bm z_\tc{a}(t_{\tc{r}}))+\tau_\tc{b}(t_{\tc{r}}, \bm z_\tc{b}(t_{\tc{r}}))$ and $s(\mu,\nu)$ is the sign function defined in Lemma~\ref{lemma:FNreversible}. Moreover, recall that the detectors' initial states are expressed as \mbox{$\ket{\psi_j}=\cos\alpha_j\ket{g_j}+\sin\alpha_j e^{\ii \beta_j}\ket{e_j}$}.
    
    These assumptions are sufficient to fulfill the conditions given in Symmetry \ref{sym:ReversalPtSwap}, thus fulfilling Corollary~\ref{cor:symmetriesT}.
\end{prop}
This proposition follows analogously to the previous one, by combining the conditions in Symmetry~\ref{sym:ReversalPtSwap} and the Lemmas~\ref{lemma:LambdaandMuSwap} and \ref{lemma:FNreversible}.


\begin{prop} \label{prop:finalForm3}
    Let $g_{\mu \nu}$ be the metric components in an arbitrary coordinate system $(t,\bm x)$, $\bm z_j(t)$ the spatial components of the detectors' trajectories and $(\tau_j,\widetilde{\bm x}_j)$ the Fermi normal coordinates adapted to $\bm z_j(t)$. Assume that
    \begin{enumerate}
        \item $g_{\mu\nu}(t_{\tc{r}}-t,\bm x_0-\bm x) = g_{\mu\nu}(t_{\tc{r}}+t,\bm x_0+\bm x)$,
        \item $G_\tc{f}(t_{\tc{r}}-t,\bm x_0-\bm x, t_{\tc{r}}-t',\bm x_0-\bm x')\\
        =G_\tc{f}(t_{\tc{r}}+t',\bm x_0+\bm x', t_{\tc{r}}+t,\bm x_0+\bm x)$,
        \item $\bm z_\tc{a}(t_{\tc{r}}-t)=2\bm x_0-\bm z_\tc{b}(t_{\tc{r}}+t)$,
        \item  $\widetilde{\bm x}_\tc{a}(t_{\tc{r}}, \bm x_0-\bm x) = \widetilde{\bm x}_\tc{b}(t_{\tc{r}}, \bm x_0+\bm x)$,
        \item $\Lambda_j(\tau_j,\widetilde{\bm x}_j)=\chi_j(\tau_j)F_j(\widetilde{\bm x}_j)$,
        \item $\chi_\tc{a}(\tau)=\chi_\tc{b}(\upsilon_\tc{r}-\tau)$,
        \item $F_\tc{a}(\widetilde{\bm x})=F_\tc{b}(\widetilde{\bm x})$,
        \item $\Omega_\tc{a} = \Omega_\tc{b}=\Omega$,
        \item $\beta_\tc{a}+\beta_\tc{b}=\Omega\upsilon_\tc{r}$, 
        \item $|\cos \alpha_\tc{a}|=|\cos \alpha_\tc{b}|$, $|\sin \alpha_\tc{a}|=|\sin \alpha_\tc{b}|$,
    \end{enumerate}
    where $\upsilon_\tc{r} = \tau_\tc{a}(t_{\tc{r}}, \bm x_0+\bm z_\tc{a}(t_{\tc{r}}))+\tau_\tc{b}(t_{\tc{r}}, \bm x_0-\bm z_\tc{b}(t_{\tc{r}}))$. Moreover, recall that the detectors' initial states are expressed as $\ket{\psi_j}=\cos\alpha_j\ket{g_j}+\sin\alpha_j e^{\ii \beta_j}\ket{e_j}$.
    
    These assumptions are sufficient to fulfill the conditions given in Symmetry \ref{sym:ReversalPtSwapReflect}, thus fulfilling Corollary~\ref{cor:symmetriesT}.
\end{prop}
This following proposition follows analogously to the previous ones, by combining the conditions in Symmetry~\ref{sym:ReversalPtSwapReflect} and the Lemmas~\ref{lemma:LambdaandMuSwapReflect} and \ref{lemma:FNreversibleReflect}.

\section{Additional symmetries when swapping the times at which the detectors are switched}
\label{apx:plotSymmetries}
In the numerical results of figures \ref{Non_symmetric_alpha=2.35} and \ref{different_detectors_1D} we have observed symmetries with respect to changing the sign of $\Delta t = t_\tc{b}-t_\tc{a}$. Here we show why these symmetries are expected.


\begin{prop}\label{prop:equallife}
    Assume that
    \begin{enumerate}
        \item The spacetime is flat.
        \item The detectors are inertial and comoving, and start in the ground state.
        \item $G_\tc{f}(t,\bm x, t',\bm x') = G_\tc{f}(t',\bm x, t,\bm x')$, in the comoving frame of reference.
        \item The switching functions are equal besides a time shift, $\chi_j(t) = \chi(t-t_j)$ (in the comoving frame of reference, where $\Lambda_j(t, \bm x)=\chi_j(t)F_j(\bm x)$).
        \item Equal gaps, $\Omega_\tc{a}=\Omega_\tc{b}$.
    \end{enumerate}
    Then,
    \begin{align}
        \mathcal M=\mathcal M_\tc{S},\quad \mathcal M^\pm=\mathcal M^\pm_\tc{S},
    \end{align}
    where $\tc{S}$ denotes swapping $t_\tc{a}$ and $t_\tc{b}$. 
\end{prop}
\begin{proof}
Using the assumptions to simplify  the Eqs.~\eqref{eq:MGen} and~\eqref{M_harvesting}, and defining $\Omega=\Omega_\tc{a}=\Omega_\tc{b}$,
\begin{align}
        \mathcal M &=-\lambda^2 \int\dd t \dd{\bm x} \dd{t'} \dd {\bm x'} e^{\ii \Omega (t + t')}\chi(t-t_\tc{a})\chi(t'-t_\tc{b})\nonumber\\
        &\qquad\qquad\times F_\tc{a}(\bm x)F_\tc{b}(\bm x')G_{\tc{f}}(t,\bm x,t', \bm x')\nonumber\\
        & = \mathcal M_\tc{S},
    \end{align}
    where, to obtain the last equality, we performed a change of variables $t \to t'$, $t' \to t$ and used the hypotheses over and $G_\tc{f}$.

    The same result holds for $\mathcal M^\pm$, because the only difference is using $G_\tc{f}^\pm$, which have the same symmetries as $G_\tc{f}$. 
\end{proof}
This Proposition \ref{prop:equallife} implies that the quantities in the plots of Figure \ref{Non_symmetric_alpha=2.35} must be even functions along the $\Delta t$ axis, which is exactly the case. The condition over $G_\tc{f}$ is automatically fulfilled for the vacuum of Minkowski, as seen in Lemma~\ref{lemma:MinkowskiGF}.

Finally, the following proposition applies to the scenario of Figure \ref{different_detectors_1D}, causing the odd symmetry that can be seen in the plot of the relative phase (after adding $\pi/2$ to the relative phase). In such plot, the swap indicated by $\tc{S}$ in what will be Eqs.~\eqref{eq:swaptimes} is equivalent to changing the sign of $\Delta t = t_\tc{b}-t_\tc{a}$. The condition over $G_\tc{f}$ is automatically fulfilled for the vacuum of Minkowski, as seen in Lemma~\ref{lemma:MinkowskiGF}.

\begin{prop} \label{prop:swapMpm} 
Assume that
\begin{enumerate}
    \item The spacetime is flat.
    \item The detectors are inertial and comoving, and start in the ground state.
    \item $\Lambda_j(t, \bm x)=\chi_j(t)F_j(\bm x)$ in the comoving frame $(t,\bm x)$.
    \item $\chi_j(t_j+t)=\chi_j(t_j-t)$.
    \item  $G_\tc{f}(t,\bm x, t',\bm x') = G_\tc{f}(t_\tc{a}+t_\tc{b}-t,\bm x, t_\tc{a}+t_\tc{b}-t',\bm x')$.
\end{enumerate}
Then, 
\begin{align}
    |\mathcal M^\pm| = |\mathcal M^\pm_\tc{S}|,\ \cos \Delta \gamma = -\cos\, (\Delta\gamma)_\tc{S}\label{eq:swaptimes},
\end{align}
where $\tc{S}$ indicates swapping $t_\tc{a}$, $t_\tc{b}$, and $\Delta \gamma$ the relative phase between $\mathcal M^+$ and $\mathcal M^-$.
\end{prop}

\begin{proof}
    Using the assumptions to simplify the Eqs.~\eqref{eq:MGen} and~\eqref{M_harvesting}, the definition of $\mathcal M^\pm$, and that switching $t_\tc{a}$ and $t_\tc{b}$ changes $\chi_{\tc{a}}(t)$ into $\chi_{\tc{a}}(t+t_\tc{a}-t_\tc{b})$ and $\chi_{\tc{b}}(t)$ into $\chi_{\tc{b}}(t+t_\tc{b}-t_\tc{a})$,
    \begin{align}
        \mathcal M^\pm_\tc{S}&=-\lambda^2 \int\dd t \dd{\bm x} \dd{t'} \dd {\bm x'} e^{\ii( \Omega_{\tc{a}}t+\Omega_{\tc{b}} t')}G_{\tc{f}}^\pm(t,\bm x, t',\bm x')\nonumber\\
        &\qquad\times \chi_{\tc{a}}(t+t_\tc{a}-t_\tc{b})\chi_{\tc{b}}(t'+t_\tc{b}-t_\tc{a}) F_\tc{a}(\bm x)F_\tc{b}(\bm x')\nonumber\\
        &=-\lambda^2 e^{\ii (\Omega_\tc{a}+\Omega_\tc{b})(t_\tc{a}+t_\tc{b})} \int\dd t \dd{\bm x} \dd{t'} \dd {\bm x'} e^{-\ii( \Omega_{\tc{a}}t+\Omega_{\tc{b}} t')}\nonumber\\
        &\qquad\times\chi_{\tc{a}}(t) \chi_{\tc{b}}(t') F_\tc{a}(\bm x)F_\tc{b}(\bm x')G_{\tc{f}}^\pm(t,\bm x, t',\bm x')\nonumber\\
        &= e^{\ii (\Omega_\tc{a}+\Omega_\tc{b})(t_\tc{a}+t_\tc{b})}\mathcal M^\pm(-\Omega_\tc{a},-\Omega_\tc{b}).
    \end{align}
    Here, to get the second equality, we performed the change of variables $t\to t_\tc{a}+t_\tc{b}-t$, $t'\to t_\tc{a}+t_\tc{b}-t'$ and used the hypotheses upon $\chi_j$ and $G_\tc{f}$.     
    
    Finally, we get the Eq.~\eqref{eq:swaptimes} by applying Proposition~\ref{prop:transformation} for the case where the initial states of the detector are the ground states, and thus $\mathcal M^\pm(-\Omega_\tc{a},-\Omega_\tc{b})=\widetilde{\mathcal M}^\pm$.
\end{proof}

\bibliography{references}

\end{document}